\documentclass[12pt]{article}
\usepackage[english]{babel}
\RequirePackage[colorlinks,citecolor=blue,urlcolor=blue,linkcolor=blue]{hyperref}
\hypersetup{
colorlinks = true,
citecolor=blue,
urlcolor=blue,
linkcolor=blue,
pdfauthor = {Jianxi Su, Edward Furman},
pdfkeywords = {Multivariate distributions, (tail) dependence, singularity, generalized Archimedean copula, Marshall-Olkin copula, credit risk, factor
models},
pdftitle = {Multiple risk factor dependence structures: Related copulas and applications in default risk},
pdfpagemode = UseNone
}
\usepackage{authblk}
\usepackage{amsmath,amsthm}
\usepackage{amsfonts}
\usepackage{graphicx}
\usepackage{enumerate}
\usepackage[usenames,dvipsnames]{color}
\usepackage{subfigure}
\usepackage{color}
\usepackage{apacite}
\usepackage{bbm}
\usepackage{epstopdf}
\usepackage[utf8]{inputenc}
\usepackage[english]{babel}

\numberwithin{equation}{section}

\makeatletter
\def\@seccntformat#1{\@ifundefined{#1@cntformat}%
   {\csname the#1\endcsname\quad}  
   {\csname #1@cntformat\endcsname}
}
\let\oldappendix\appendix 
\renewcommand\appendix{%
    \oldappendix
    \newcommand{\section@cntformat}{\appendixname~\thesection\quad}
}
\makeatother

\usepackage{natbib}
\setcitestyle{authoryear,open={(},close={)}}

\usepackage{filecontents}

    \def\qed{\hfill$\sqcap\kern-8.0pt\hbox{$\sqcup$}$\\}

\newtheorem{theorem}{Theorem}
\newtheorem{lemma}{Lemma}

\newtheorem{proposition}{Proposition}
\newtheorem{corollary}{Corollary}
\theoremstyle{definition}
\newtheorem{definition}{Definition}
\newtheorem{example}{Example}

 \date{}

\begin{document}

\noindent
{\Large \bf Multiple risk factor dependence structures: Copulas and related properties}

\vspace*{8mm}
\noindent
{\large Jianxi Su$^{a}$, Edward Furman$^{*,b}$}

\noindent
$^a$ Department of Statistics, Purdue University, West Lafayette, IN 47907, United State

\noindent
$^b$ Department of Mathematics and Statistics, York University, Toronto, ON M3J 1P3, Canada

\noindent
\rule{165mm}{0.2mm}
\vspace{-10mm}
\begin{quote}
\textbf{Abstract.}
Copulas have become an important tool in the modern best practice Enterprise Risk Management, often supplanting other approaches to modelling stochastic dependence. However, choosing the `right' copula is not an easy task, and the temptation to prefer a tractable rather than a meaningful candidate from the encompassing copulas toolbox is strong. The ubiquitous applications of the Gaussian copula is just one illuminating example.

Speaking generally,  a `good' copula should conform to the problem at hand, allow for asymmetry in the domain of definition and exhibit some extent of tail dependence. In this paper we introduce and study a new class of Multiple Risk Factor (MRF) copula functions, which we show are exactly such. Namely, the MRF copulas  (1) arise from a number of meaningful default risk specifications with stochastic default barriers,  (2) are in general non-exchangeable and (3) possess a variety of tail dependences. That being said, the MRF copulas turn out to be surprisingly tractable analytically.

\bigskip



\vspace*{12mm}
\noindent
\rule{100mm}{0.2mm}
\\
{\footnotesize
$^{*}$Corresponding author. Tel. +1(416)736-2100 ext. 33768. \\
E-mail addresses: jianxi@purdue.edu (J. Su), efurman@mathstat.yorku.ca (E. Furman).
}

\end{quote}

\newpage

\section{Introduction}
\label{sec-1}

Copulas are beautiful mathematical constructions, and as such they have become a well established quantitative tool in actuarial and financial research and practice (e.g., Denuit et al., 2005; McNeil et al., 2005; and references therein). However, with the tractability comes a price. Namely, while we must choose a copula depending on the problem at hand, this choice is somewhat vague. As a result, practitioners often choose copulas due to the mathematical convenience, rather than because of meaningful connections to the phenomena they model.  The reason is that such connections are frequently very difficult to find. Luckily there are exceptions.

To set off, we recall that an $n(\in\mathbf{N})$-variate function
$C : [0, 1]^n\rightarrow  [0, 1]$ is a copula, if it is grounded, $n$-increasing, and have uniformly distributed margins (e.g., Joe, 1997;  Nelsen, 2006).
\begin{example}[Marshall-Olkin (MO) copula (e.g., Cherubini et al., 2013)]
\label{ex-1}
Consider spouses that purchase a life insurance. The joint future lifetime of the couple can be modelled by the random variable (r.v.) $(\tau_1,\ \tau_2)'\in\mathbf{R}_{0,+}^2:=[0,\infty)^2$, such that
\begin{equation}
\label{sr-mo-1}
\tau_1={}_1E_{\lambda_1}\wedge E_{\lambda_0} \textnormal{ and } \tau_2={}_2E_{\lambda_2}\wedge E_{\lambda_0},
\end{equation}
where ${}_1E_{\lambda_1}$, ${}_2E_{\lambda_2}$ and $E_{\lambda_0}$ are all exponentially distributed and stochastically independent r.v.'s having positive scale parameters  $\lambda_1,\ \lambda_2$ and $\lambda_0$, respectively (Bowers et al., 1997).

Stochastic representation (\ref{sr-mo-1}) is quite natural, as the spouses may die either at independent future times ${}_1E_{\lambda_1}$ and ${}_2E_{\lambda_2}$ - as a result of the individual mortality, or simultaneously (fully comonotonically) - as a result of a joint fatal hazard (common shock). The joint survivorship probability of the future lifetimes is then (Marshall and Olkin, 1967)
\[
\mathbf{P}[\tau_1>s,\ \tau_2>t]=\exp\{-\lambda_1s-\lambda_2t-\lambda_0(s\vee t)\},
\textnormal{ where }s,t\in \mathbf{R}_{0,+},
\]
and a routine application of Sklar's theorem (Sklar, 1959) yields that the corresponding copula function (e.g., Cherubini et al., 2013) is
\[
C(u,v)=u^{1-\lambda_0/(\lambda_1+\lambda_0)}v\wedge
v^{1-\lambda_0/(\lambda_2+\lambda_0)}u,
\textnormal{ where } u,v\in [0,\ 1].
\]
In summary, the (bivariate) MO copula can be mapped to a stochastic representation that describes a meaningful real world phenomenon of interest to actuaries.
\end{example}

\begin{example}[Clayton copula (Clayton, 1978)]
\label{ex-2}
Consider two risk components in a portfolio of default risks, and let the coordinates of the r.v. $(\tau_1,\ \tau_2)'\in\mathbf{R}_{0,+}^2$ denote the default times of these risk components. Furthermore, let ${}_1E_{\lambda}$ and ${}_2E_{\lambda}$ be two exponentially distributed r.v.'s that are independent mutually as well as on a gamma distributed r.v. $\Lambda$ having shape parameter $\gamma(\in\mathbf{R}_{+})$ and unit scale parameter. Last but not least, denote by `$\ast$' the mixture operator, such that ${}_1E_\Lambda{\overset{d}{=}}{}_1E_\lambda\ast\Lambda$ and ${}_2E_\Lambda{\overset{d}{=}}{}_2E_\lambda\ast\Lambda$, where `$\overset{d}{=}$' denotes equality in distribution.
Then we may be interested in the following default specification
\begin{equation}
\label{sr-cl-1}
\tau_{1}={}_1E_{\Lambda} \textnormal{ and }
\tau_{2}={}_2E_{\Lambda}.
\end{equation}

Stochastic representation (\ref{sr-cl-1}) is a simplification of the CreditRisk$^+$ approach to modelling the risk of default (Bielecki and Rutkowski, 2004), and it is easy to see that the corresponding joint survival function is (e.g., Albrecher et al., 2011; Su and Furman, 2016a)
\[
\mathbf{P}[\tau_1>s,\ \tau_2>t]=\left(
1+s+t
\right)^{-\gamma},
\textnormal{ where }
s,t\in\mathbf{R}_{0,+}.
\]
Moreover, the obtained dependent times of occurrence (hitting times) $(\tau_1,\ \tau_2)'$ are positively quadrant dependent (PQD) (Lehmann, 1966),
and Sklar's theorem yields the following copula (e.g., Joe, 1997; Nelsen, 2006)
\[
C(u,v)=\left(u^{-\gamma}+v^{-\gamma}-1\right)^{-1/\gamma},
\textnormal{ where } u,v\in[0,\ 1].
\]
Hence, similarly to the case of the Marshall-Olkin copula in Example \ref{ex-1}, the Clayton copula admits a stochastic representation that is of interest to (credit) risk professionals.
\end{example}


The goal of this paper is to introduce and study a class of copula functions that unify, among others, the MO and Clayton copulas discussed in Examples \ref{ex-1} and \ref{ex-2}, respectively. More specifically, the Multiple Risk Factor (MRF) copulas introduced herein admit meaningful stochastic representations, are non-exchangeable and allow for a significant variety of tail dependences, and nevertheless are surprisingly tractable analytically. Immediate areas of application of the MRF copulas are life insurance and default risk management. E.g., in the latter context, the MRF dependencies describe default risk portfolios, which are exposed to an arbitrary number of fatal risk factors having conditionally exponential hitting times that can be independent, positively orthant dependent (POD) (Lehmann, 1966) and even fully comonotonic (Dhaene et al., 2002a,b).

The rest of the paper is organized as follows.  In Section \ref{sec-con} we introduce the MRF copulas in their most general form along with the various links to default specifications having stochastic default barriers. One of the interesting peculiarities of the MRF copulas is the fact that they are not absolutely continuous with respect to the Lebesgue measure, thus allowing for a non-zero probability of simultaneous default.  We study the phenomenon of simultaneous default generally in Section \ref{sec-sing}, and we specialize the discussion to the context of the Clayton subclass of the MRF copulas in Section \ref{third-con}, where we also study the dependence properties of the Clayton MRF copulas thoroughly. Last but not least, we explore the extremal (tail) dependence behaviour of the Clayton MRF copulas in Section \ref{forth-con}, where we employ both the classic indices of tail dependence and the new notion of maximal tail dependence introduced recently in Furman et al. (2015). Section \ref{sec-concl}
concludes the paper. The proofs are relegated to the appendix.

\section{Construction of the multiple risk factor copula functions and some basic properties}
\label{sec-con}

Consider a risk portfolio (r.p.) that consists of $n$ risk components (r.c.'s) with the labels in the set $\{1,\ldots,n\}$. Let $\boldsymbol{X}=(X_1,\ldots,X_n)'$ denote a r.v. with the $i$-th coordinate interpreted as the default time of the $i$-th r.c. with $i=1,\ldots,n$, and assume that each r.c. is exposed to some of (or all of) $(l+m)$ fatal risk factors (r.f.'s) of which $l(\in\mathbf{N})$ r.f.'s have fully-comonotonic hitting times and $m(\in\mathbf{N})$ r.f.'s have  POD hitting times. Further, let the block matrix $c=(c^l,\ c^m)\in Mat_{n\times (l+m)}(\mathbf{1})$ have entries in $\mathbf{1}:=\{0,\ 1\}$ and describe the exposure of the r.p. $\{1,\ldots,n\}$ to the distinct r.f.'s in the set $\{1,\ldots,l+m\}$; we assume that the matrix $c$ is deterministic and may in practice be chosen by the senior risk management.
Finally, let the sets $\mathcal{RF}^l_i=\{j\in\{1,\ldots,l\}:\ c^l_{i,j}=1\}$, $\mathcal{RF}^m_i=\{j\in\{l+1,\ldots,l+m\}:\ c^m_{i,j}=1\}$ and $\mathcal{RF}_i=\mathcal{RF}_i^l \cup \mathcal{RF}_i^m$ contain the r.f.'s that `hit' the $i$-th r.c., $i=1,\ldots,n$. Similarly, denote by $\mathcal{RC}_j=\{i\in\{1,\ldots,n\}:\ c_{i,j}=1\}$ the set that contains all the r.c.'s that are hit by the $j$-th r.f., $j=1,\ldots,l+m$.

To make the distributional structure underlying the general set-up above tractable analytically, we assume hereafter that 
\begin{itemize}
\item  [(A1)] for a fixed r.f. in the sets $\{1,\ldots,l\}$ and $\{l+1,\ldots,l+m\}$, the hitting time r.v.'s are conditionally fully-comonotonic and conditionally independent, respectively;
\item [(A2)] the r.v. $\boldsymbol{\Lambda}:=(\Lambda_1,\ldots,\Lambda_{l+m})'$ gathers the uncertainty about r.f.'s, and the coordinates $\Lambda_1,\ldots,\Lambda_{l+m}$ are mutually independent stochastically;
\item  [(A3)] for varying r.f.'s in the set $\{1,\ldots,l+m\}$, the hitting time r.v.'s are stochastically independent and distributed exponentially, succinctly $E_{\lambda_1},\ldots,E_{\lambda_l}$ and
${}_iE_{\lambda_{l+1}},\ldots,{}_iE_{\lambda_{l+m}}$, given $\Lambda_1=\lambda_1,\ldots, \Lambda_{l+m}=\lambda_{l+m}$, where $i=1,\ldots,n$, and $\lambda_1,\ldots,\lambda_{l+m}$ are all positive.
\end{itemize}

We have already mentioned the notion of mixture operator (Example \ref{ex-2}). More specifically, given two appropriately jointly measurable r.v.'s $X_\beta\sim C(\beta)$ with $\beta\in\mathcal{B}\subseteq \mathbf{R}$ and $B\sim H$, the `mixture' r.v. $X_B$  has the same distribution as $X_\beta\ast B$, where the r.v. $B$ has its range in $\mathcal{B}$.
Then, for $i=1,\ldots,n$,
let $N_{\lambda_jt},\ j=1,\ldots,l$ and ${}_iN_{\lambda_jt}, \ j=l+1,\ldots,l+m$  denote stochastically independent homogeneous
Poisson processes with intensities $\lambda_j$ such that
$\mathbf{P}[N_{\lambda_jt}=0]=\mathbf{P}[E_{\lambda_j}>t]$ and
$\mathbf{P}[{}_iN_{\lambda_j t}=0]=\mathbf{P}[{}_iE_{\lambda_j}>t]$ for any $t\in  \mathbf{R}_{0,+}$. Finally, set
\begin{eqnarray}
\label{MRF-poisson}
\tau_i=\inf \left\{t\in \mathbf{R}_{0,+}:
\sum_{j\in\mathcal{RF}^l_i} N_{\Lambda_jt}+\sum_{j\in\mathcal{RF}^m_i} {}_iN_{\Lambda_jt}>0 \right\}
\end{eqnarray}
to represent the default time of the $i$-th r.c., where
$\Lambda_j,\ j\in\mathcal{RF}_i$ are positive r.v.'s and $i=1,\ldots,n$.
Given assumptions (A2) and (A3) above,
it is easy to show that, for $t\in \mathbf{R}_{0,+}$, the marginal survival probability of  $\tau_i$ is
\begin{equation}
S_i(t):=\mathbf{P}[\tau_i> t]
=\psi_{{\sum_{j\in\mathcal{RF}_i} \Lambda_j}}(t),
\label{marginal-general-ddf}
\end{equation}
where $\psi_{{\sum_{j\in\mathcal{RF}_i} \Lambda_j}}(t)=\mathbf{E}[e^{-\sum_{j\in\mathcal{RF}_i} \Lambda_j t}]$; here and throughout $\psi_X(x)$ denotes the Laplace transform of the r.v. $X$ evaluated at $x\in\mathbf{R}_{0,+}$.
In a similar fashion and with a bit of an effort, we show that, for $t_i\in \mathbf{R}_{0,+},\
i=1,\ldots,n$ and
$\boldsymbol{\Lambda}$ as before, the joint survival probability is given by
\begin{equation}
S(t_1,\ldots,t_n):=\mathbf{P}[\tau_1>t_1,\ldots,\tau_n>t_n]\nonumber\\
=
\prod_{j=1}^{l}\psi_{\Lambda_j}\left(\bigvee_{i\in\mathcal{RC}_j}t_i\right)
\prod_{j=l+1}^{l+m}\psi_{\Lambda_j}\left(\sum_{i\in\mathcal{RC}_j}t_i\right).
\label{multivariate-general-ddf}
\end{equation}

In practice, the mixed (doubly stochastic) Poisson processes that generate defaults must not be necessarily homogeneous. Namely, we may be interested in the integrated intensities
$\Lambda_j(t),\ t\in \mathbf{R}_{0,+},\ j=1,\ldots,l+m$, which are real valued, continuous and increasing stochastic processes
such that $\Lambda_j(0)=0$. As a result (\ref{MRF-poisson}) can be generalized to
\begin{equation}
\label{Non-hom-defaults}
\tau_i=\inf \left\{t\in \mathbf{R}_{0,+}: \sum_{j\in \mathcal{RF}^l_i} \ N_{\Lambda_{j}(t)}+\sum_{j\in \mathcal{RF}_i^m}\ {}_iN_{\Lambda_{j}(t)}>0 \right\},\
\end{equation}
where  $i=1,\ldots,n$. The survival function of $\tau_i$ as well as the joint survival function of $(\tau_1,\ldots,\tau_n)'$ are formulated in the next
theorem. The proofs are omitted, as they very much resemble the derivations that led to (\ref{marginal-general-ddf}) and
(\ref{multivariate-general-ddf}).

\begin{theorem}
\label{pro-survial probability}
For default specification (\ref{Non-hom-defaults}) and assuming that
$\Lambda_j(t),\ t\in \mathbf{R}_{0,+}$ are real valued,  continuous and increasing stochastic processes
with $\Lambda_j(0)=0,\ j=1,\ldots,l+m$, the marginal survival probability of the $i$-th r.c. is given by
\begin{eqnarray}
S_i(t)= \psi_{\sum_{j\in \mathcal{RF}_i}\Lambda_{j}(t)}(1) \label{uniddf},
\end{eqnarray}
for $i=1,\ldots,n$.
Also, the joint survivorship probability of the risk portfolio $\{1,\ldots,n\}$ is formulated as
\begin{eqnarray}
S(t_1,\ldots,t_n)=
\prod_{j=1}^{l}\psi_{\Lambda_{j}(\bigvee_{i\in\mathcal{RC}_j}t_i)}(1) \prod_{j=l+1}^{l+m} \psi_{\sum_{i\in\mathcal{RC}_j}\Lambda_{j}(t_i)}(1) \label{jointddf},
\end{eqnarray}
where $t_i\in \mathbf{R}_{0,+},\ i=1,\ldots,n$.
\end{theorem}

We next show that the general form of the MRF dependencies, and hence  (\ref{Non-hom-defaults}), admit the so called
default specification with stochastic default barrier.
To this end, for $i=1,\ldots,n$ and $t\in\mathbf{R}_{0,+}$, let
\begin{equation}\label{Theta}
\Theta_{i}(t)=\sum_{j\in \mathcal{RF}^l_i}\mathbf{I}^{\infty}_{\{N_{\Lambda_j(t)}>0\}}+\sum_{j\in \mathcal{RF}^m_i}\Lambda_j(t),
\end{equation}
where (Jacod and Shiryaev, 2003)
\[
\mathbf{I}^{\infty}_{\{N_{\Lambda_j(t)}>0\}}:=\left\{
                                     \begin{array}{ll}
                                       0, & \hbox{$N_{\Lambda_j(t)}=0$} \\
                                       \infty, & \hbox{$N_{\Lambda_j(t)}>0$}
                                     \end{array}
                                   \right..
\]

\begin{theorem}
\label{pro-alternative def1}
Let ${}_iE_1\sim Exp(1)$ and $U_i\sim Uni[0,\ 1],\ i=1,\ldots,n$ denote stochastically independent r.v.'s that are, respectively,
exponentially distributed with unit scale parameters, and uniformly distributed on $[0,\ 1]$.
Then the following stochastic default specifications are equivalent mutually as well as to  (\ref{Non-hom-defaults})
\begin{itemize}
\item[(D1)]
exponential default barrier representation
\begin{eqnarray}
\label{formula-alternative def1}
\tau_i=\inf\left\{t\in\mathbf{R}_{0,+}:\ \Theta_i(t)\geq {}_iE_1 \right\};
\end{eqnarray}
\item[(D2)] uniform default barrier representation
\begin{eqnarray}
\label{formula-alternative def3}
\tau_i=\inf\{t\in\mathbf{R}_{0,+}:\ \exp\{-\Theta_i(t)\}\leq U_i \}.
\end{eqnarray}
\end{itemize}
\end{theorem}

We note in passing that default specifications with stochastic barriers \`a la (D1) have been
discussed in, e.g.,  Lando (2004), Escobar et al. (2012), Skoglund and Chen (2015) and references
therein. Interestingly, (D1) to an extent reduces  the complexity involved
in simulating the r.v.'s $\tau_1,\ldots,\tau_n$. Indeed, note that according to
(\ref{Theta}) and concentrating on the r.f.'s with POD hitting times,
we have that in order to simulate $\tau_i,\ i=1,\ldots,n$,
it is only necessary to simulate at most $m$ sample paths
of the stochastic processes
$\Lambda_j(t),\ j\in\{l+1,\ldots,l+m\}$ as well as one exponential r.v. with unit
scale, whereas the stochastic representations in Su and Furman (2016b) requires an
$n\times m$ array  of such exponentials.

On a different note,
default specification (D2) suggests that non-negative probabilities of simultaneous default in the context of (\ref{Non-hom-defaults}) can only manifest as a result of the r.f.'s in the set $\{1,\ldots,l\}$. Also, (D2) is of special interest as the next theorem hints. Let
$S_i^{-}:[0,\ 1]\rightarrow \mathbf{\overline{R}}_{0,+}:=[0,\ \infty]$ be a function, such that
\[
S^{-}_i(u):=\inf\{
x\in\mathbf{\overline{R}}_{0,+}:\ S(x)\leq u
\},
\]
where $u\in[0,\ 1]$ and $\inf\{\emptyset\}=\infty$ by convention.
The function $S_i^{-}$ is called the generalized inverse of $S_i$, and as such it is equal to the usual inverse $S^{-1}_i(u)$ if the survival function
is continuous (Embrechts and Hofert, 2013). The proof of the following theorem is omitted, as it is a direct consequence of Sklar's theorem (Sklar, 1959) and Theorem \ref{pro-survial probability}.

\begin{theorem}
\label{MRFcopulaTh}
The copula function $C:[0,\ 1]^n\rightarrow [0,\ 1]$ that corresponds to the general MRF dependence structures is given,
 for $u_i\in[0,\ 1],\ i=1,\ldots,n$, by
\begin{eqnarray}
\label{formula-sibuya copula}
C(u_1,\ldots,u_n)=\prod_{j=1}^{l} \psi_{\Lambda_{j}(\bigvee_{i\in \mathcal{RC}_j}S^{-}_i(u_i))}(1)\prod_{j=l+1}^{l+m} \psi_{\sum_{i\in \mathcal{RC}_j} \Lambda_{j}(S^{-}_i(u_i))}(1),
\end{eqnarray}
where $S^{-}_i(u_i)$ is the (generalized) inverse of $S_i(t)=\psi_{\sum_{j\in\mathcal{RF}_i} \Lambda_j(t)}(1)$ and
$t\in\mathbf{R}_{0,+}$.
\end{theorem}

MRF copulas (\ref{formula-sibuya copula})  are well-tailored to model dependent
default times or, more generally, dependent risks in the context of the Enterprise Risk Management (ERM).
Indeed, the MRF copulas
emerge from default time specifications (\ref{Non-hom-defaults}), (\ref{formula-alternative def1}) and
(\ref{formula-alternative def3}), cover the full range of non-negative dependence when it is measured by the Spearman rho measure of correlation (Section \ref{third-con}) and
are non-exchangeable unless the exposure matrix $c$ is such that $c_{1,j}=\cdots=c_{n,j}$ for all $j=1,\ldots,l+m$.
Furthermore, the MRF copulas reduce  to the product copula
and the Fr\'echet upper bound copula, if there are only
idiosyncratic r.f.'s $j\in\{1,\ldots,l+m:\ \sum_{i\in\mathcal{RC}_j}=1\}$
and  only
systemic r.f.'s
$j\in\{1,\ldots,l:\ \sum_{i\in\mathcal{RC}_j^l}=n\}$, respectively, included.

Notwithstanding,
the MRF copulas in their most general form are somewhat too abstract to be tackled
analytically. As it often happens, some simplifying assumptions are necessary.
For instance, it is possible to consider a class of linear stochastic processes
$\Lambda_j(t)=\Lambda_jt,\ j=1,\ldots,l+m$, only. In such a case, (\ref{formula-sibuya copula}) yields the
following class of copula functions
\begin{eqnarray}\label{cop-lin}
C(u_1,\ldots,u_n)&=&\prod_{j=1}^{l} \psi_{\Lambda_j}\left
(\bigvee_{i\in \mathcal{RC}_j}\psi^{-1}_{\sum_{j\in \mathcal{RF}_i} \Lambda_j}(u_i)\right)
\prod_{j=l+1}^{l+m}\psi_{\Lambda_{j}}\left(\sum_{i\in \mathcal{RC}_j}\psi^{-1}_{\sum_{j\in \mathcal{RF}_i} \Lambda_j}(u_i)\right)\nonumber,\\
\end{eqnarray}
where  $u_i\in [0,1]$ and
$i=1,\ldots,n$.  A simplification of (\ref{formula-sibuya copula}), the subclass of the MRF copulas  in (\ref{cop-lin}) is rich enough to
unify the well-known Archimedean and Marshall-Olkin classes of copula functions. In fact,
we have that  (\ref{cop-lin}) introduces a class of non-exchangeable
 Archimedean copulas and recovers the class of the Marshall-Olkin copulas, if
the sets $\mathcal{RC}_j$ contain at least two elements for some
$j\in\{l+1,\ldots,l+m\}$ and $j\in\{1,\ldots,l\}$, respectively. Moreover, we have that (\ref{cop-lin}) simplifies to the product copula
if the sets $\mathcal{RC}_j$ have at most one element for all $j\in\{1.\ldots,l+m\}$, and it reaches the Fr\' echet upper  bound copula
if the cardinalities of the sets $\mathcal{RC}_j,\ j\in\{1,\ldots,l\}$ coincide with the dimension of the copula whereas these sets are empty
for all other risk factors.

The following theorem establishes a characteristic representation of  MRF copulas (\ref{cop-lin}) \`a la the popular common-shock framework (e.g., Asimit et al., 2010; Su and Furman, 2016b).
\begin{theorem}
\label{pro-simulation}
For $i=1,\ldots,n$, let $V_j,\ j=1,\ldots, l$ and ${}_iV_j,\ j=l+1,\ldots,l+m$
denote a sequence of independent uniform $U[0,1]$ r.v.'s. Then the r.v. $\mathbf{U}=(U_1,\ldots,U_n)'$ has cumulative distribution function (c.d.f.) (\ref{cop-lin}) if and only if
\begin{eqnarray}
\label{stochastic representation}
U_i=\left(\bigvee_{j\in \mathcal{RF}^l_i} \psi_i \left(-\frac{\ln(V_{j})}{\Lambda_j} \right)\right)
\bigvee
\left(\bigvee_{j\in \mathcal{RF}_i^m} \psi_i \left(-\frac{\ln({}_iV_j)}{\Lambda_j} \right) \right).
\end{eqnarray}
\end{theorem}

We conclude this section with some references. Namely, we note that the class of Archimedean copulas has been extensively used
in the context of credit risk in, e.g., Hull
and White (2006), Choudhry (2010), Constantinescu et al. (2011) and references
therein, as well as in the general ERM in, e.g., Frees and Valdez (1998), Kole et al. (2007), Sandstr\"{o}m (2010) and Staudt (2010). Also, the class of the Marshall-Olkin copulas has been recently
suggested for applications in credit risk in Cherubini et al. (2013), and its applications
to insurance mathematics were presented in Bowers et al. (1997).

\section{Probability of simultaneous default}
\label{sec-sing}

Survival function (\ref{jointddf}) is not absolutely continuous with respect to the Lebesgue measure on $\mathbf{R}_{0,+}^n$, and as a result,
default specifications (\ref{MRF-poisson}),  (\ref{Non-hom-defaults}), (\ref{formula-alternative def1}) and
(\ref{formula-alternative def3}) as well as copula functions (\ref{formula-sibuya copula}) and (\ref{cop-lin})
allow for non-zero probabilities of simultaneous
default. One empirical justification for accommodating this phenomenon is the famous example of $24$ railway
firms defaulting on June 21, 1970 (Azizpour and Giesecke, 2008), another justification, that is somewhat more theoretical,
is the conclusion of Das et al. (2007) that
the mixed Poisson processes approach tends to underestimate the clustering of real world
defaults.

We next formulate the probability of simultaneous default for the general MRF dependencies
discussed in the previous section.
To this end and for $2\leq k\leq n$,
we denote the
set of all risk factors that `attack' the sub-portfolio
$\{i_1,\ldots,i_k\}$ by
\begin{equation}
\label{RFall}
\mathcal{RF}_{i_1,\ldots,i_k}=\left\{j\in\{1,\ldots,l+m\}:\ c_{i_h,j}=1 \textnormal{ for at least one } i_h\in\{i_1,\ldots,i_k\}\right\},
\end{equation}
and we note that it is the union of two disjoint sets, that is of
\begin{equation}
\label{RFsing}
\mathcal{RF}_{(i_1,\ldots,i_k)}:=\{j\in\{1,\ldots,l+m\}:  c_{i_h,j}=1
\textnormal{ for all } i_h\in\{i_1,\ldots,i_k\}\}
\end{equation}
and
\begin{equation}
\label{RFnotsing}
\mathcal{RF}_{\overline{(i_1,\ldots,i_k)}}:=\mathcal{RF}_{i_1,\ldots,i_k}\setminus
\mathcal{RF}_{(i_1,\ldots,i_k)}.
\end{equation}
In addition, for $1\leq h\leq k$ and $2\leq k\leq n$, we let
\begin{equation}
\label{RFihnotsing}
\mathcal{RF}_{{i_h},\overline{(i_1,\ldots,i_k)}}=\mathcal{RF}_{i_h}\setminus \mathcal{RF}_{(i_1,\ldots,i_k)}.
\end{equation}
Last but not least, for $t\in \mathbf{R}_{0,+}$, let
\begin{equation}
\label{Aset}
\mathcal{A}(t)=\left\{
\sum_{j\in \mathcal{RF}^l_{(i_1,\ldots,i_k)}}N_{\Lambda_j(t)}>0\ \textnormal{ and }
\sum_{j\in \mathcal{RF}^l_{(i_1,\ldots,i_k)}}N_{\Lambda_j(t-)}=0
\right\}.
\end{equation}
\begin{theorem}
\label{pro-simultaneous probability}
Consider default specification (\ref{Non-hom-defaults}), and let $\{i_1,\ldots,i_k\}$ establish an index set with
$2\leq k\leq n$,  then the probability of simultaneous default is given by
\begin{eqnarray}
&&\mathbf{P}[\tau_{i_1}=\cdots=\tau_{i_k}]
=\int_{\mathbf{R}_{0,+}}\mathbf{P}[\mathcal{A}(t)]\prod_{j\in\mathcal{RF}^l_{\overline{(i_1,\ldots,i_k)}}}
\psi_{\Lambda_j(t)}(1)\prod_{j\in \mathcal{RF}^m_{i_1,\ldots,i_k}}\psi_{\Lambda_j(t) }(| \mathcal{RC}_j | ) dt,\nonumber\\
\label{thsim}
\end{eqnarray}
where $|\cdot|$ denotes set's cardinality.
\end{theorem}

Under an additional assumption of linearity of the integrated intensity, the
probability of simultaneous default can be simplified as follows.

\begin{corollary}
\label{lemma-simultaneous default}
Let $\Lambda_j(t)=\Lambda_j t$ for all $j=1,\ldots,l+m$ and $t\in\mathbf{R}_{0,+}$, and leave the rest of the set-up in Theorem
\ref{pro-simultaneous probability} unchanged, then
the probability of simultaneous default of the sub-portfolio $\{i_1,\ldots,i_k\},\ 2\leq k\leq n$ is given by
\begin{eqnarray}\label{corsim}
\mathbf{P}[\tau_{i_1}=\cdots=\tau_{i_k}]=\mathbf{E}\left[
\frac{\Lambda^l_{(i_1,\ldots,i_k)}}{\Lambda^l_{i_1,\ldots,i_k}+
\tilde{\Lambda}^m_{i_1,\ldots,i_k}}
\right],
\end{eqnarray}
where $\Lambda^l_{(i_1,\ldots,i_k)}= \sum_{j\in\mathcal{RF}^l_{(i_1,\ldots,i_k)}}\Lambda_j$, $\Lambda^l_{i_1,\ldots,i_k}= \sum_{j\in\mathcal{RF}^l_{i_1,\ldots,i_k}}\Lambda_j$, $\tilde{\Lambda}^m_{i_1,\ldots,i_k}=\sum_{j\in\mathcal{RF}^m_{i_1,\ldots,i_k}}\Lambda_j  | \mathcal{RC}_j|$
and $|\cdot|$ stands for set's cardinality.
\end{corollary}

Unless very special cases are of interest, e.g., when the set of the r.f.'s having fully-comonotonic hitting times is empty
($\mathcal{RF}^l_{(i_1,\ldots,i_k)}=\emptyset$)
or when the set of the r.f.'s having POD hitting times is empty ($\mathcal{RF}^m_{(i_1,\ldots,i_k)}=\emptyset$) and in addition
the r.v.'s $\Lambda_j,\ j\in\mathcal{RF}_{i_1,\ldots,i_k}$ follow favourable probability distributions (Section \ref{third-con}),
even expression (\ref{corsim}) is somewhat involved to be handled analytically. However, it is worth noting that 
Corollary \ref{lemma-simultaneous default} is quite convenient to compute the probability of simultaneous default employing Monte-Carlo simulations.

In summary, we have hitherto derived a number of important results in the context of the general form of
the MRF dependencies as well as in the special case when the integrated intensities
are linear. However, in order to obtain insights into such higher level characteristics of the
new dependence structures as, e.g., measures of rank correlation and indices of tail dependence, further assumptions are required. In the following, we assume that the r.v.'s
$\Lambda_j$ are stochastically independent and distributed gamma with shape parameters
$\xi_j(\in\mathbf{R}_{+})$ and unit scales, $j=1,\ldots,l+m$.
The above choice of the distribution of $\Lambda_j$ may seem ad
hoc at the first glance, but it is well motivated by the CreditRisk$^+$ approach, which serves
as one of the most popular ways to model default risk in the modern practice of credit risk.
The assumption gives birth to the Clayton subclass of the MRF  dependencies.

\section{Clayton multiple risk factor copula functions}
\label{third-con}
Let $\Lambda_j\sim Ga(\xi_j,\ 1),\ j=1,\ldots,l+m$ denote $(l+m)$ stochastically
 independent r.v.'s distributed gamma. Then the probability density function of $\Lambda_j$ is
 \begin{eqnarray}
\label{gamma-law}
f_{\Lambda_j}(\lambda;\xi_j,1)=e^{-\lambda}\frac{\lambda^{\xi_j-1}}
{\Gamma(\xi_j)},\ \lambda\in\mathbf{R}_{+},
\end{eqnarray}
and the corresponding Laplace transform is
\begin{equation}\label{LTg}
\psi_{\Lambda_j}(x)=\left(
1+x
\right)^{-\xi_j},\ x\in\mathbf{R}_{0,+}.
\end{equation}
The latter observation immediately establishes that gamma distributions are infinitely divisible and so
closed under convolutions, i.e., in our case, we have that $\Lambda:=\Lambda_1+\cdots
+\Lambda_n$ is distributed $Ga(\xi,1)$, where $\xi=\xi_1+\cdots+\xi_n$. Importantly,
even if the scale parameters are not equal, the distribution of the convolution is still a gamma but with
a random shape parameter. This is formulated in the following lemma.

Let $K$ be an integer valued non-negative r.v. with the probability mass function (p.m.f.) $p_k:=\mathbf{P}[K=k]$, which is given by
\begin{equation}
\label{pk}
p_k=c_{+} \delta_k,\ k=0,\ 1, \ldots,
\end{equation}
where, for $\sigma_i\in\mathbf{R}_{+}$ and $\sigma_+=\vee_{i=1}^n \sigma_i$,
\[
c_{+}=\prod_{i=1}^n \left(\frac{\sigma_i}{\sigma_{+}}\right)^{\xi_i}
\]
and
\[
\delta_k=\left\{
\begin{array}{ll}
1, & k=0 \\
k^{-1}\sum_{l=1}^k\sum_{i=1}^n \xi_i \left(
1-\frac{\sigma_i}{\sigma_{+}}
\right)^l\delta_{k-l}, & k=1,\ 2,\ldots
\end{array}
\right..
\]
\begin{lemma}[Moschopoulos, 1985, also, H\"urlimann, 2001 and Furman and Landsman, 2005]
\label{FL2005}
Let $\Lambda_i\sim Ga(\xi_i(\in\mathbf{R}_{+}),\ \sigma_i(\in\mathbf{R}_{+})),\ i=1,\ldots,n$ denote gamma distributed and independent
stochastically r.v.'s with arbitrary shape and scale parameters, and let $\Lambda=\Lambda_1+\cdots+\Lambda_n$ be their convolution. Then
$\Lambda\sim Ga(\xi+K,\ \sigma_+)$, where $\xi=\xi_1+\cdots+\xi_n$, $\sigma_+=\vee_{i=1}^n \sigma_i$ and $K$ is
an integer valued non-negative r.v. with p.m.f. (\ref{pk}).
\end{lemma}

We note in passing that if $\sigma_1=\cdots=\sigma_n$, then $K=0$ almost surely, and the findings of the lemma reduce to the
simple convolution of gamma distributed r.v.'s with equal scale parameters.

We further introduce the Clayton subclass of the MRF dependencies. The definition
below follows from
(\ref{cop-lin}) and (\ref{LTg}). We remind that the doubly stochastic Poisson approach
with gamma distributed intensities has been adapted in CreditRisk$^+$,
and, as such, it is arguably one of the most popular ways
to model dependent defaults in nowadays credit risk practice. The method has been often
criticized for underestimating the clustering of defaults' occurrences (Das et al., 2007).  We note that the Clayton MRF dependencies augment
the POD hitting times of
the r.f.'s in CreditRisk$^+$ with the fully comonotonic hitting times of the so called
systemic r.f.'s. This allows for a mechanism to model the clustering of defaults
more accurately and may thus resolve to an extent the aforementioned drawback of the CreditRisk$^+$ method.

\begin{definition}
\label{GCS-def}
Copula functions $C_{\boldsymbol{\xi}}:[0,1]^n\rightarrow [0,1]$, parametrized by the deterministic vector
$\boldsymbol{\xi}=(\xi_1,\ldots,\xi_{l+m})'$ with $\xi_j\in\mathbf{R}_{+},\ j=1,\ldots,l+m$,
are called the Clayton MRF copulas if
\begin{eqnarray}
\label{ddf-copula}
C_{\boldsymbol{\xi}}(u_1,\ldots,u_n)
&=&\prod_{j=1}^{l}\bigwedge_{i\in \mathcal{RC}_j}u_i^{\frac{\xi_j}{\xi_{c,i}}}\prod_{j=l+1}^{l+m}   \left[1+\sum_{i\in \mathcal{RC}_j}\left(u_i^{-\frac{1}{\xi_{c,i}}} -1\right) \right]^{-\xi_j},
\end{eqnarray}
where $u_i\in[0, 1]$ and  $\xi_{c,i}=\sum_{j\in \mathcal{RF}_i} \xi_j$ for $i=1,\ldots,n$.
\end{definition}

In order to state our next results, we break the vector parameter $\boldsymbol{\xi}=(\xi_1,\ldots,\xi_{l+m})'$ as following $\boldsymbol{\xi}:=(\boldsymbol{\alpha}',\boldsymbol{\gamma}')'$ where $\boldsymbol{\alpha}:=(\alpha_1,\ldots,\alpha_l)'$ and $\boldsymbol{\gamma}:=(\gamma_{l+1},\ldots,\gamma_{l+m})'$. Then, with the help of the set notations introduced earlier, we can have general sums of the form $\square_{\bullet}=\sum_{j\in\mathcal{RF}_{\bullet}} \square_j$, where `$\square$' can be a parameter, e.g., $\alpha,\ \gamma,\ \xi$ or a r.v., e.g., $\Lambda$, and `$\bullet$' is any one of $i_1,\ldots,i_k$, $(i_1,\ldots,i_k)$, $(\overline{i_1,\ldots,i_k})$ and $i_h,(\overline{i_1,\ldots,i_k)}$.

It is easy to see that, for a fixed dimension and $(u_1,\ldots,u_n)'\in[0,\ 1]^n$,
\begin{itemize}
\item if the sets $\mathcal{RC}_j$ contain at most one element for all $j\in\{1,\ldots,l+m\}$,
then 
\[
C_{\boldsymbol{\xi}}(u_1,\ldots,u_n)=\prod_{i=1}^n u_i=:C^\perp(u_1,\ldots,u_n)
\textnormal{ - the product copula};
\]
\item if $|\mathcal{RC}_j|=n$ for some $j\in\{1,\ldots,l\}$ and are zero otherwise, then
\[
C_{\boldsymbol{\xi}}(u_1,\ldots,u_n)=\bigwedge_{i=1}^n u_i=:M(u_1,\ldots,u_n)
\textnormal{ - the Fr\'echet upper bound  copula};
\]
\item if $|\mathcal{RC}_j|=n$ for some $j\in\{l+1,\ldots,m+l\}$ and are zero otherwise, then
\begin{eqnarray*}
&&C_{\boldsymbol{\xi}}(u_1,\ldots,u_n)=\left[1+\sum_{i=1}^n\left(u_i^{-\frac{1}{\gamma_{c,(1,\ldots,n)}}} -1\right) \right]^{-\gamma_{c,(1,\ldots,n)}}=:C_{\boldsymbol{\gamma}}(u_1,\ldots,u_n)
\\
&& \textnormal{- the Clayton Archimedean  copula};
\end{eqnarray*}
\item if the sets $\mathcal{RC}_j$ have at least one element for $j\in\{1,\ldots,l\}$
and are empty sets otherwise,
then
\[
C_{\boldsymbol{\xi}}(u_1,\ldots,u_n)=\prod_{j=1}^{l}\bigwedge_{i\in \mathcal{RC}_j}u_i^{\frac{\alpha_j}{\alpha_{c,i}}}
=:C_{\boldsymbol{\alpha}}(u_1,\ldots,u_n)
\textnormal{ - the Marshall-Olkin copula}.
\]
\end{itemize}

The closure under convolutions property of the r.v.'s distributed gamma, and more generally Lemma
\ref{FL2005}, facilitate yet additional simplification of (\ref{thsim}) in the context of the Clayton MRF copulas. 
A very special case of Theorem \ref{pro-simultaneous probability} and Corollary \ref{lemma-simultaneous default}, the following
proposition establishes at a stroke two stand alone results obtained independently in Marshall and Olkin (1967) and Asimit et al. (2010).
\begin{proposition}
\label{CGS simultaneous default}
Within the Clayton subclass of the MRF copulas and with $\{i_1,\ldots,i_k\},\ 2\leq k\leq n$ establishing an index set, the probability of
simultaneous default is given by
\begin{eqnarray}
\mathbf{P}\left[U_{i_1}^{1/\xi_{c,i_1}}=\cdots=U_{i_k}^{1/\xi_{c,i_k}}\right]=\alpha_{(i_1,\ldots,i_k)}\mathbf{E}\left[\frac{1}{\xi_{c,i_1,\ldots,i_k} + K}\right],
\end{eqnarray}
where $\alpha_{c,(i_1,\ldots,i_k)}=\sum_{j\in \mathcal{RF}^l_{(i_1,\ldots,i_k)}} \xi_j$, $\xi_{c,i_1,\ldots,i_k}=\sum_{j\in \mathcal{RF}_{i_1,\ldots,i_k}} \xi_j$, and $K$ is an integer-valued r.v. having p.m.f. \`a la (\ref{pk}).
\end{proposition}

The rest of this section is devoted to deriving the Spearman rho measure of rank correlation
in the context of the Clayton subclass of the MRF dependencies. It is well-known that
the Pearson measure of correlation can produce somewhat counter-intuitive results
when the dependence is not linear, i.e., beyond the class of multivariate elliptical
distributions (see, e.g., Fang et al., 1990). The Spearman rho, succinctly $\rho_S$, provides a
natural extension for arbitrary dependencies.

\begin{definition}[Nelsen, 2006]
Let the r.v.'s $U$ and $V$ have a copula $C$.  Then the Spearman rho measure of
rank correlation is given by
\[
\rho_S(C)=12\int \int_{[0,\ 1]^2}uvdC(u,v)-3.
\]
\end{definition}

A number of notes are instrumental before formulating the expression for
$\rho_S$ in the context of the Clayton MRF dependencies. First, we are interested in the
bivariate copula functions only and thus (\ref{ddf-copula}) reduces to
\begin{eqnarray}
\label{bivariate-ddf}
C_{\boldsymbol{\xi}}(u_i,u_k)
&=&u_i^{\xi_{c,i,\overline{(i,k)}}\over\xi_{c,i}}u_k^{\xi_{c,k,\overline{(i,k)}}\over\xi_{c,k}}
\left(
u_i^{\alpha_{c,(i,k)}\over \xi_{c,i}}\bigwedge
u_k^{\alpha_{c,(i,k)}\over \xi_{c,k}}
\right)
\left(u_i^{-\frac{1}{\xi_{c,i}}}+u_k^{-\frac{1}
{\xi_{c,k}}}-1 \right)^{-\gamma_{c,(i,k)}},
\end{eqnarray}
where $u_i$ and $u_k$ are in $[0,\ 1]$ for $1\leq i\neq k\leq n$.

Second, we recall that the $(q+1)\times q$ hypergeometric function
(Gradshteyn and Ryzhik, 2014) is formulated as
\begin{eqnarray}
\label{hyperpq}
_{q+1}F_q(a_1,\ldots,a_{q+1};b_1,\ldots,b_q;z):=\sum_{k=0}^{\infty}\frac{(a_1)_k,\ldots,(a_{q+1})_k }{(b_1)_k,\ldots,(b_q)_k}\times\frac{z^k}{k!},\
\end{eqnarray}
where $(p)_n:=p(p+1)\cdots(p+n-1)$ for $n\in \mathbf{Z}_+$, $(p)_0:=1$ and $q\in\mathbf{Z}_+$.
For $a_1,\ldots,a_{q+1}$ all positive, and these are the cases of interest in the present paper,
the radius of convergence of the series is the open disk $|z|<1$. On the boundary $|z|=1$, the series
converges absolutely if $d=b_1+\cdots + b_q-a_1-\cdots -a_{q+1}>0$, and it
converges except at $z=1$ if $0\geq d>-1$.

Let
\[
h(x)={}_3F_2(2x,1,\gamma_{c,(i,k)};2x+1,2\xi_{c,i}+2\xi_{c,k}-\xi_{c,(i,k)}+1;-1),
\]
where $x\in\mathbf{R}_{0,+}$, and such that $h(x)$ is well defined.
\begin{theorem}
\label{sp-rho}
Consider the Clayton subclass of the MRF copulas, then the
Spearman measure of rank correlation is, for $1\leq i\neq k\leq n$, given by
\begin{eqnarray}
\label{sp-rho-ClaytonMRF-formula}
\rho_S(C_{\boldsymbol{\xi}})=\frac{6}{2\xi_{c,i}+2\xi_{c,k}-\xi_{c,(i,k)}}(\xi_{c,k}h(\xi_{c,i})+\xi_{c,i}h(\xi_{c,k}))-3.
\end{eqnarray}
\end{theorem}

Two immediate consequences are formulated next. We note in passing that while
Corollary \ref{MO-spearman-thm}
confirms the findings in Embrechts et al. (2003), Corollary
\ref{pro-as-spearman} is seemingly new.
\begin{corollary}
\label{MO-spearman-thm}
Let $\gamma_{c,(i,k)}\equiv 0$  and leave the rest of the conditions in Theorem \ref{sp-rho}
unchanged. In this case, the Clayton MRF copula reduces to the
Marshall-Olkin copula, succinctly $C_{\boldsymbol{\alpha}}$, with the measure of Spearman rank correlation  given by
 \begin{eqnarray}
\label{marshall-olkin spearman}
\rho_S(C_{\boldsymbol{\alpha}})=\frac{3\alpha_{c,(i,k)}}{2\xi_{c,i}+2\xi_{c,k}-\alpha_{c,(i,k)}}.
\end{eqnarray}
\end{corollary}

\begin{corollary}
\label{pro-as-spearman}
Let $\alpha_{c,(i,k)}\equiv 0$ and leave the rest of the conditions in Theorem \ref{sp-rho}
unchanged. In this case, the Clayton MRF copula reduces to the class of non-exchangeable
Archimedean copulas, succinctly $C_{\boldsymbol{\gamma}}$, with the measure of Spearman rank correlation  given by
\begin{eqnarray}
\rho_S(C_{\boldsymbol{\gamma}})=3\left[\ _3F_2\left(1,1,\gamma_{c,(i,k)};2\xi_{c,i}+1,2\xi_{c,k}+1;1 \right) -1\right].
\end{eqnarray}
\end{corollary}

In the next section we study the dependence of extreme default times, i.e., the tail dependence, of the
Clayton subclass of the MRF copulas. As the majority of the existing methods for quantifying tail dependence aim at
random pairs, we specialize the discussion in the next section to the bivariate case, only.
Some of our following results can be extended to the multivariate case with just a bit
of an effort, others are rather involved when explored in higher dimensions and can serve
as great future research topics for a technically adept mathematician.

\section{Tail dependence of the generalized Clayton copula}
\label{forth-con}
Speaking plainly, tail dependence is about the clustering of
extreme events. In the context of default risk, such clustering is written formally as
\begin{eqnarray*}
\mathbf{P}[\tau_i\leq t,\tau_k\leq t]
=1-\mathbf{P}[\tau_i > t]-\mathbf{P}[\tau_k > t]+\mathbf{P}[\tau_i > t,\tau_k>t],
\end{eqnarray*}
for $1\leq i\neq k \leq n$ and $t\uparrow t^\ast$, where $t^\ast$ is the maximal
default time of the risk portfolio $\{i,k\}$. In view of the above and keeping in mind
that the Clayton MRF copulas are in fact survival copulas, i.e., they couple survival functions,
in what follows we restrict our attention to the copula $C_{\boldsymbol{\xi}}(u,\ v):=
\mathbf{P}[(U,V)\in R(u,v)]$, where the
rectangle $R(u,v(:=[0,\ u]\times [0,\ v]$ `shrinks' along the diagonal
$\{(u,v)\in[0,\ 1]:\ u=v\}$ (Subsection \ref{cl-tail-dep-subsec}) or along a more intricate path $(\varphi(u),\psi(u))_{0\leq u,v\leq 1}$, where $u\downarrow 0$, and $\varphi$ and $\psi$ are eligible functions (Subsection \ref{max-tail-dep-suubsec}).

\subsection{Classical measures of tail dependence}
\label{cl-tail-dep-subsec}

Speaking generally, there exist a variety of ways to quantify the extent of tail dependence
in bivariate random vectors with dependence structures gathered by copulas (Nelsen, 2006; Durante and Sempi, 2015).
Arguably the most popular
measure of lower tail dependence is nowadays attributed to Joe (1993) and given by
\begin{eqnarray}
\label{lambda}
\lambda_L:=\lambda_L(C)= \lim_{u \downarrow 0} {C(u,u)\over u}.
\end{eqnarray}
Non-zero (more precisely $(0,\ 1]$) values of (\ref{lambda})
suggest lower tail dependence in $C$.

On a different note,
when limit (\ref{lambda}) is zero, it is often useful to turn to the somewhat more
delicate index of weak tail dependence $\chi_L\in[-1,\ 1]$ (Coles et al., 1999; Fischer and Klein, 2007)
that is given by
\begin{equation}\label{ind-ccl}
\chi_{L}:=\chi_L(C)=\lim_{u\downarrow 0}\frac{2\log u }{\log C(u,u)}-1 ,
\end{equation}
and/or to the index of intermediate tail dependence $\kappa_L:=\kappa_L(C)\in[1,\ 2]$ (Ledford and
Tawn, 1996; Hua and Joe, 2011) that solves the equation
\begin{equation}\label{ind-kkl}
C(u,u)= \ell(u) u^{\kappa_L} \quad \textrm{when} \quad u\downarrow 0,
\end{equation}
assuming that we can find a slowly varying at $0$ function $\ell(u)$.

We next compute indices (\ref{lambda}), (\ref{ind-ccl}) and (\ref{ind-kkl}) in the context of the Clayton subclass of the MRF copulas.  We recollect to this end, that similarly to the general MRF
copulas, the Clayton MRF copula functions admit default specifications with the exogenous
r.f.'s having
stochastically independent hitting times (idiosincratic r.f.'s) and positively orthant dependent or even fully comonotonic
hitting times (systemic r.f.'s).
\begin{proposition}
\label{pro-strong tail dependence}
In the context of the Clayton subclass of the MRF dependencies, we have, for $1\leq i\neq k\leq n$,
 that
\[
\lambda_L(C_{\boldsymbol{\xi}})=\left\{
                                                    \begin{array}{ll}
                                                      0, & \hbox{$\xi_{c,i,\overline{(i,k)}}\neq 0$ and/or $\xi_{c,k,\overline{(i,k)}} \neq0$} \\
                                                      2^{-\gamma_{c,(i,k)}}, & \hbox{$\xi_{c,i,\overline{(i,k)}}=\xi_{c,k,\overline{(i,k)}}=0$.}
                                                    \end{array}
                                                  \right.
\]
\end{proposition}

Hence, the copula of the random default times $(\tau_i,\ \tau_k)'$
is lower tail dependent in the sense of (\ref{lambda})
if the underlying default specification does not include idiosyncratic
r.f.'s.  Furthermore, the higher the contribution of the systemic r.f.'s is
(higher values of $\alpha_{c,(i,k)}$ and thus, for fixed margins, lower values of
$\gamma_{c,(i,k)}$), the more lower tail dependent the copula of the default times $(\tau_i,\ \tau_k)'$
is.

\begin{proposition}
\label{pro-two indices}
Within the Clayton subclass of the MRF dependencies, we have, for $1\leq i\neq k\leq n$, that
the index of weak lower tail dependence is given by
\begin{eqnarray}
\label{xl-genCl}
\chi_L(C_{\boldsymbol{\xi}})= \frac{\xi_{c,(i,k)}}{\xi_{c,i}+\xi_{c,i,\overline{(i,k)}}}
\bigwedge \frac{\xi_{c,(i,k)}}{\xi_{c,k}+\xi_{c,k,\overline{(i,k)}}} ,
\end{eqnarray}
whereas the index of intermediate lower tail dependence is given by
\begin{eqnarray}
\label{kl-genCl}
\kappa_L(C_{\boldsymbol{\xi}})=2-\frac{\xi_{c,(i,k)}}{\xi_{c,i}}\bigwedge\frac{\xi_{c,(i,k)}}{\xi_{c,k}}.
\end{eqnarray}
\end{proposition}

Indices (\ref{lambda}), (\ref{ind-ccl}) and (\ref{ind-kkl}) may underestimate the amount of tail dependence in copulas that are symmetric or asymmetric, with or without singularities (Furman et al., 2015). The reason is that all the aforementioned indices of lower tail dependence rely entirely on the behaviour of copulas along their main diagonal $(u,\ u)_{0\leq u\leq 1}$. However, the tail dependence of copulas can be substantially stronger along the paths other than the main diagonal. This can be a serious disadvantage, as reported by, e.g., Schmid and Schmidt (2007), Zhang (2008), Li et al. (2014), and Furman et al. (2015). In the next example, we elucidate this phenomenon in the
context of the Clayton subclass of the MRF dependencies.

\begin{example}
\label{ex-cl-tail}
Consider the bivariate Clayton MRF copula with $\gamma_{c,(i,k)}\equiv 0$,
$\xi_{c,i,\overline{(i,k)}}\neq 0$ and $\xi_{c,k,\overline{(i,k)}}\neq 0$. Then, by
(\ref{bivariate-ddf}),
\begin{eqnarray}
\label{MO-C}
C_{\boldsymbol{\alpha}}(u_i,u_k)=u_i^{\xi_{c,i,\overline{(i,k)}}\over\xi_{c,i}}u_k^{\xi_{c,k,\overline{(i,k)}}\over\xi_{c,k}}\left(
u_i^{\alpha_{c,(i,k)}\over \xi_{c,i}}\wedge
u_k^{\alpha_{c,(i,k)}\over \xi_{c,k}}\right),\ (u_i,u_k)'\in [0,1]^2.
\end{eqnarray}
Appealing to Propositions \ref{pro-strong tail dependence} and \ref{pro-two indices},
we readily have, for $1\leq i\neq k\leq n$, that
\[
\lambda_L(C_{\boldsymbol{\alpha}})=0
\textnormal{ and }
\kappa_L(C_{\boldsymbol{\alpha}})=2-\frac
{\alpha_{c,(i,k)}}{\xi_{c,i}}\bigwedge\frac
{\alpha_{c,(i,k)}}{\xi_{c,k}}.
\]

Then denote by $\lambda^\ast_L(C_{\boldsymbol{\alpha}})$ and
$\kappa_L^\ast(C_{\boldsymbol{\alpha}})$ two indices \`a la
(\ref{lambda}) and (\ref{ind-kkl}), respectively, but along an alternative than
diagonal path, and let such path be the singularity path $\left(u^{2\xi_{c,i}/(\xi_{c,i}+\xi_{c,k})}\right.$,
$\left.u^{2\xi_{c,k}/(\xi_{c,i}+\xi_{c,k})}\right)_{0\leq u\leq 1}$.
In this case, we readily obtain that
\[
\lambda_L^*(C_{\boldsymbol{\alpha}})=0=\lambda_L(C_{\boldsymbol{\alpha}})
\textnormal{ and }
\kappa_L^*(C_{\boldsymbol{\alpha}})=2-\frac{2\alpha_{c,(i,k)}}{\xi_{c,i}+\xi_{c,k}}
\leq \kappa_L(C_{\boldsymbol{\alpha}}),
\]
where the equality holds only if $\xi_{c,i}=\xi_{c,k}$.
\end{example}

To conclude, Example \ref{ex-cl-tail} shows that in the context of the Clayton MRF copulas, all classic indices of tail dependence may not yield the maximal measures of extreme default times' co-movements.

\subsection{Measures of maximal tail dependence}
\label{max-tail-dep-suubsec}
If there existed a one word paradigm that could characterize the modern regulatory accords in
financial risk management, then it would be `prudence'. Indeed, regulators around the globe
have been making tremendous efforts to convey the necessity of modelling the effect
of `low probability/high severity risks' on the risk portfolios of insurance companies and banks.
We next formally introduce measures of maximal tail dependence. To this end, we heavily borrow from
Furman et al. (2015).

\begin{definition}
\label{paths}
A function $\varphi: [0,1] \to [0,1]$ is called {\it admissible} if it satisfies the following conditions:
\begin{enumerate}
\item [(C1)]
$\varphi(u) \in [u^2,1] $ for every $u\in [0,1] $; and
\item [(C2)]
$\varphi(u)$ and $u^2/\varphi(u) $ converge to $0$ when $u\downarrow 0$.
\end{enumerate}
Then the path $(\varphi(u),u^2/\varphi(u))_{0\le u \le 1}$ is admissible whenever the function $\varphi $ is admissible. Also, we denote by $\mathcal{A}$ the set of all admissible
functions $\varphi$.
\end{definition}

A number of observations are instrumental to clarify the definition. First,
condition (C1) makes sure that both $\varphi(u)\in[0,\ 1]$ and
$u^2/\varphi(u) \in [0,\ 1]$, whereas  condition (C2) is motivated by the fact that we are
interested in the behavior of the copula $C$
near the lower-left vertex of its domain of definition. Second,
it is clear that the function $\varphi_0(u)=u,\ u\in[0,\ 1]$ is admissible and yields the main
diagonal $(u,\ u)_{0\leq u\leq 1}$. Last but not least, for the independence
copula, it holds that $C^{\perp}(\varphi(u),u^2/\varphi(u))=u^2,\  u\in[0,\ 1]^2$,
which is path-independent as expected, thus warranting the choice
$\zeta(u)=u^2/\varphi(u),\ u\in [0,\ 1]$.

In order to determine the strongest extreme co-movements of risks
for any copula $C$, we search for functions $\varphi\in\mathcal{A}$ that maximize the probability
\[
\Pi_{\varphi }(u) =C\big (\varphi(u),u^2/\varphi(u)\big ), \;\;\; u \in (0,1)
\]
or, equivalently, the function
\[
d_\varphi(C,C^\perp)(u)=C\big (\varphi(u),u^2/\varphi(u)\big )-
C^\perp (\varphi(u),u^2/\varphi(u)\big ), \;\;\; u \in (0,1),
\]
which is non-negative for PQD copulas $C$.
Then an admissible function $\varphi^* \in \mathcal{A} $ is called {\it a function of maximal dependence} if
\begin{equation}
\label{Pi-m}
\Pi_{\varphi^* }(u)=\max_{\varphi \in \mathcal{A} }\Pi_{\varphi }(u)
\end{equation}
for all $u\in (0,1)$. The corresponding admissible path
$(\varphi^*(u),u^2/\varphi^*(u))_{0\leq u\leq 1}$ is called
{\it a path of maximal dependence}.
Generally speaking, the path $\varphi^\ast$ is not unique, but for each such path the value of $\Pi_{\varphi^\ast}$ is the same. In what follows, we use the notation $\Pi^*(u)$ instead of $\Pi_{\varphi^* }(u)$.

Prudent variants of measures
(\ref{lambda}), (\ref{ind-ccl}) and (\ref{ind-kkl}) are then introduced as
\begin{equation}\label{lambdaast}
\lambda_L^*:=\lambda_L^*(C)= \lim_{u \downarrow 0} {\Pi^*(u)\over u},
\textnormal{ instead of }
\lambda_{L}(C)=\lim_{u \downarrow 0} \frac{C(u,u)}{u}
\end{equation}
and
\begin{equation}
\label{chaiast}
\chi_L^*:=\chi_L^*(C)=\lim_{u\downarrow 0}\frac{2\log u }{\log \Pi^*(u)}-1,
\textnormal{ instead of }
\chi_L=\lim_{u\downarrow 0}\frac{2\log u }{\log C(u,\ u)}-1,
\end{equation}
subject to the existence of the limits, and also
\begin{equation}
\label{kappaast}
\Pi^*(u)=\ell^*(u) u^{\kappa_L^*},\ u\downarrow 0,
\textnormal{ as opposed to }
\Pi(u)=\ell(u) u^{\kappa_L},\ \quad u\downarrow 0,
\end{equation}
assuming that there exist slowly varying at zero functions $\ell^\ast(u)$ and $\ell(u)$.

A useful technique for deriving function(s) of maximal dependence, and thus in turn of the corresponding
indices, consists of three
steps:
\begin{enumerate}
  \item [(S1)] search for critical points of the function $x \mapsto C(x,u^2/x)$ over the interval $[u^2,1]$ and for each $u\in [0,\ 1]$;
  \item [(S2)] check which of the solution(s) is/are global maximum/maxima; and
  \item [(S3)] verify that the function $u\mapsto \varphi^*(u)$  is admissible.
\end{enumerate}

\begin{figure}[h!]
\centering
\includegraphics[width=7.5cm,height=7.5cm]{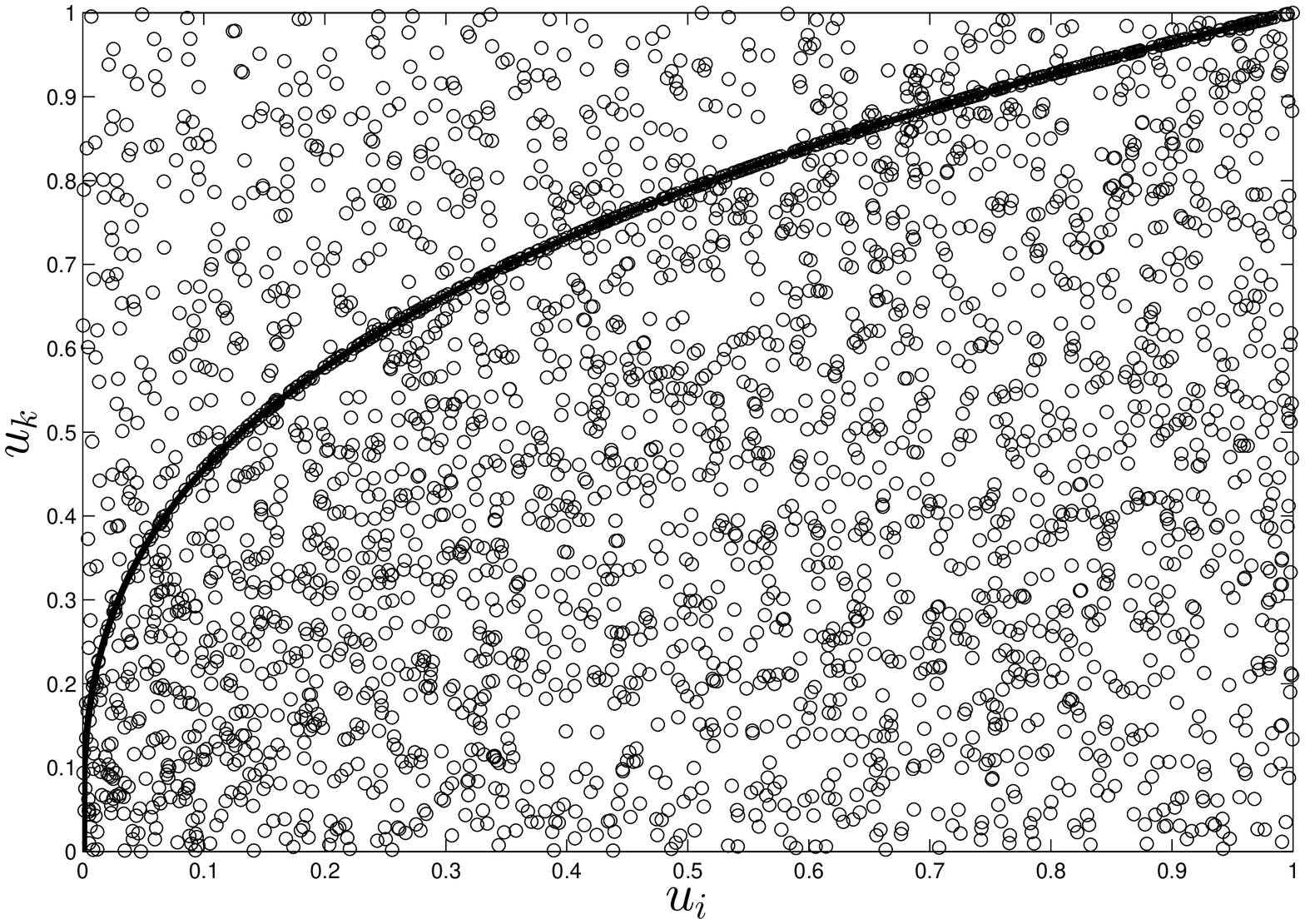}
\hfill
\includegraphics[width=7.5cm,height=7.5cm]{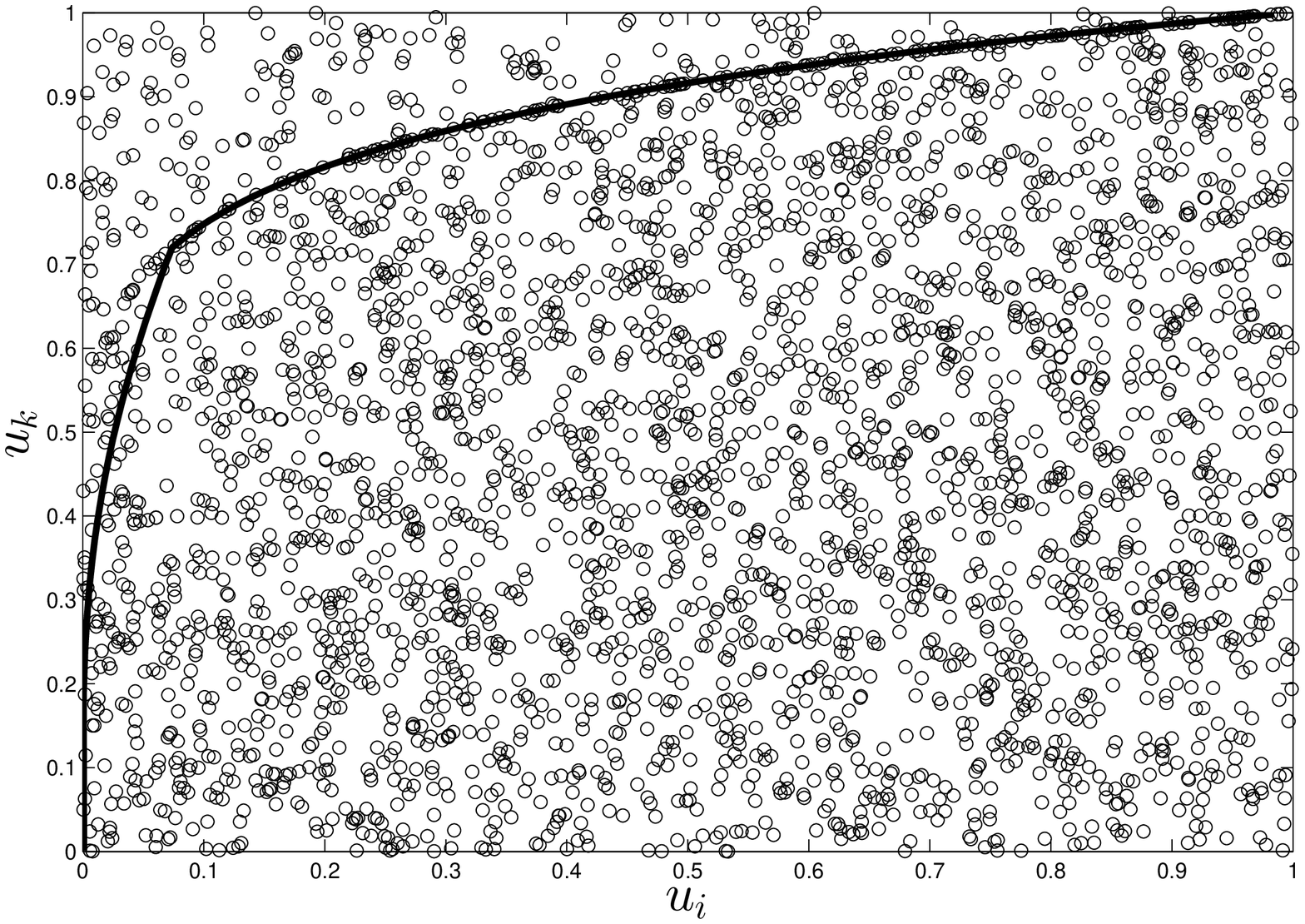}
\caption{Scatterplots of  the Clayton MRF copula for $\xi_{c,i,\overline{(i,k)}}=3,\ \xi_{c,k,\overline{(i,k)}}=0.3,\ \gamma_{c,(i,k)}=0.5,\  \alpha_{c,(i,k)}=0.6$ (left panel) and $\xi_{c,i,\overline{(i,k)}}=10,\ \xi_{c,k,\overline{(i,k)}}=0.3,\ \gamma_{c,(i,k)}=0.5,\   \alpha_{c,(i,k)}=0.6$ (right panel) with the paths of maximal dependence superimposed on both panels.}
\label{fig-mdp-simulation}
\end{figure}

We next formulate and prove the main result of this section. Figure \ref{fig-mdp-simulation}
visualizes some of the notions in it.
\begin{theorem}
\label{pro-general mdp}
Consider the Clayton subclass of the MRF copulas, then we have, for $1\leq i\neq k\leq n$,
 that
\begin{itemize}
\item the index of maximal strong lower tail dependence is given by
\[
\lambda^*_L(C_{\boldsymbol{\xi}})=\lambda_L(C_{\boldsymbol{\xi}});
\]
\item the index of maximal weak lower tail dependence is given by
\[
\chi^*_{L}(C_{\boldsymbol{\xi}})=\frac{\xi_{c,(i,k)}}{\xi_{c,i,\overline{(i,k)}}+\xi_{c,(i,k)}+\xi_{c,k,\overline{
(i,k)}}};
\]
\item the index of  maximal intermediate lower tail dependence is given by
\[
\kappa^*_L(C_{\boldsymbol{\xi}})=2\left(
1-\frac{\xi_{c,(i,k)}}{\xi_{c,i,\overline{(i,k)}}+2\xi_{c,(i,k)}+\xi_{c,k,\overline{(i,k)}}}
\right).
\]
\end{itemize}
\end{theorem}


\section{Conclusions}
\label{sec-concl}

Copulas have become an important element of the best practice ERM, superseding in many
contexts other more traditional approaches to modelling stochastic dependence.
However, choosing the right copula is not an easy call, and the temptation to make the
decision in
favour of a tractable rather than a meaningful copula is high. The use of the
Gaussian copula to price the collateralized debt obligations is one illuminating example.

A good copula should conform to a problem at hand, be asymmetric and exhibit some extent
of tail dependence. The MRF copulas that we have introduced and studied in this
paper are exactly such. Namely, they admit  stochastic
representations that are meaningful in the context of the ERM, arise from a number of default risk specifications with stochastic default barriers and are in general not symmetric in their domains of definition. Furthermore, the MRF copulas cover the full range of non-negative dependence when measured by the Spearman rho index of rank correlation, allow for a variety of tail dependences, and are yet quite tractable analytically.

Among  immediate applications,
the MRF copulas
generalize the CreditRisk$^+$ framework by augmenting systemic risk factors having
fully comonotonic hitting times, thus making the overall default times more positively orthant
dependent. As the CreditRIsk$^+$ method has been criticized for underestimating the
clustering of real world defaults, we believe that the MRF copulas may serve as a reasonable
supplement. That being said, as the notion of systemic risk is of fundamental importance
in the nowadays era of prudent risk management, we think that the MRF copulas may be
of interest for modelling general dependent (insurance) risks, well beyond the context of credit
risk.

\section*{Acknowledgements}
We are grateful to Prof. Dr. Paul Embrechts and all participants of the ETHs
Series of Talks in Financial and Insurance Mathematics for feedback and insights.

Our research has been supported by the Natural Sciences and Engineering Research Council (NSERC) of Canada. Jianxi Su also acknowledges the financial support of the Government of Ontario and MITACS Canada via, respectively, the Ontario Graduate Scholarship program
and the Elevate Postdoctoral fellowship.

\section*{References}



\hangindent=\parindent\noindent Albrecher, H., Constantinescu, C., Loisel, S., 2011. Explicit ruin formulas for models with dependence among risks.  {Insurance: Mathematics and Economics}   {48}(2), 265--270.

\hangindent=\parindent\noindent Asimit, A.V., Furman, E., Vernic, R., 2010. On a multivariate Pareto distribution.
  {Insurance: Mathematics and Economics}   {46}(2), 308--316.

\hangindent=\parindent\noindent Azizpour, S., Giesecke, K., 2008. Self-exciting corporate default: Contagion vs. frailty. Technical report, Stanford University, Stanford.


\hangindent=\parindent\noindent Bielecki, T.R., Rutkowski, M., 2004.   {Credit Risk: Modeling, Valuation and Hedging.}
Spinger, Berlin.

\hangindent=\parindent\noindent Bowers, N.L., Gerber, H.U., Hickman, J.C., Jones, D.A., Nesbitt, C.J., 1997.
  {Actuarial Mathematics}, 2nd ed. Society of Actuaries, Schaumburg.


\hangindent=\parindent\noindent Cherubini, U., Durante, F., Mulinacci, S. (Eds.), 2013.   {Marshall-Olkin Distributions -
Advances in Theory and Applications}. Springer, Bologna.


\hangindent=\parindent\noindent Choudhry, M., 2010.   {Structured Credit Products: Credit Derivatives and Synthetic Securitisation}, 2nd ed. Wiley, Singapore.

\hangindent=\parindent\noindent Constantinescu, C., Hashorva, E., Ji, L., 2011.   {Archimedean copulas in finite and infinite dimensions - with application to ruin problems.} {Insurance: Mathematics and Economics} {49}(3), 487--495.

\hangindent=\parindent\noindent Clayton, D.G., 1978.   A model for association in bivariate life tables and its application in epidemiological studies of familial tendency in chronic disease incidence. Biometrika 65(1), 141--151. 

\hangindent=\parindent\noindent Coles, S., Heffernan, J., Tawn, J., 1999. Dependence measures for extreme value
analyses.   {Extremes}   {2}(4), 339--365.

\hangindent=\parindent\noindent Das, S.R., Duffie, D., Kapadia, N., Saita, L., 2007. Common failings: How corporate
defaults are correlated.   {Journal of Finance}   {62}(1), 93--117.

\hangindent=\parindent\noindent Denuit, M., Dhaene, J., Goovaerts, M., Kass, R., 2006.. Actuarial Theory for Dependent
Risks: Measures, Orders and Models, Wiley.

\hangindent=\parindent\noindent Dhaene, J., Denuit, M., Goovaerts, M.J., Kaas, R., Vyncke, D., 2002a. The concept of comonotonicity in actuarial science and finance: theory. {Insurance: Mathematics and Economics} {31}(1), 3--33.

 \hangindent=\parindent\noindent Dhaene, J., Denuit, M., Goovaerts, M.J., Kaas, R., Vyncke, D., 2002b. The concept of comonotonicity in actuarial science and finance: applications. {Insurance: Mathematics and Economics} {31}(2), 133--161.


\hangindent=\parindent\noindent Durante, F., Sempi, C., 2015.   {Principles of Copula Theory.} CRC Press, Boca Raton.


\hangindent=\parindent\noindent Embrechts, P., Hofert, M., 2013. A note on generalized inverses. {Mathematical Methods of Operations Research}   {77}(3), 423--432.


\hangindent=\parindent\noindent Escobar, M., Arian, H., Seco, L., 2012. CreditGrades framework within stochastic covariance models.   {Journal of Mathematical Finance}   {2}(4), 303--313.

\hangindent=\parindent\noindent Fang, K-T., Kotz, S., Ng, K-W., 1990.   {Symmetric Multivariate and Related Distributions}. CRC Press, Boca Raton.

\hangindent=\parindent\noindent Fischer, M., Klein, I., 2007. Some results on weak and strong tail dependence coefficients for means of copulas. Technical Report 78/2007, Friedrich-Alexander-Universitat Erlangen-Nurnberg, Lehrstuhl fur Statistik und Okonometrie.




\hangindent=\parindent\noindent Frees, E.W., Valdez, E., 1998. Understanding relationships using copulas. {North American Actuarial Journal} {2}(1), 1-25. 

\hangindent=\parindent\noindent Furman, E., Kuznetsov, A., Su, J., Zitikis, R., 2016. Tail dependence of the Gaussian copula revisited.   {Insurance: Mathematics and Economics} 69, 97--103.

\hangindent=\parindent\noindent Furman, E., Landsman, Z., 2005. Risk capital decomposition for a multivariate dependent gamma portfolio.   {Insurance: Mathematics and Economics}   {37}(3), 635--649.

\hangindent=\parindent\noindent Furman, E., Su, J., Zitikis, R., 2015. Paths and indices of maximal tail dependence.
  {ASTIN Bulletin}   {45}(3), 661--678.


\hangindent=\parindent\noindent Gradshteyn, I.S., Ryzhik, I.M., 2014.   {Table of Integrals, Series, and Products}, 8th ed.
Academic Press, New York.

\hangindent=\parindent\noindent Hua, L., Joe, H., 2011. Tail order and intermediate tail dependence of multivariate copulas.   {Journal of Multivariate Analysis}   {102}(10), 1454--1471.

\hangindent=\parindent\noindent Hull, J.C., White, A.D., 2006. Valuing credit derivatives using an implied copula
approach.   {The Journal of Derivatives}   {14}(2), 8--28.

\hangindent=\parindent\noindent H\"urlimann, W., 2001. Analytical evaluation of economic risk capital for portfolios of gamma risks. {ASTIN Bulletin} {31}(1), 107--122.

\hangindent=\parindent\noindent Jacod, J., Shiryaev, A., 2003.   {Limit Theorems for Stochastic Processes}. Springer, Berlin.

\hangindent=\parindent\noindent Joe, H., 1993. Parametric families of multivariate distributions with given margins.
  {Journal of Multivariate Analysis}   {46}(2), 262--282.

\hangindent=\parindent\noindent Joe, H., 1997.   {Multivariate Models and Dependence Concepts.} CRC Press, Boca Raton.


\hangindent=\parindent\noindent Kole, E., Koedijk, K., Verbeek, M., 2007. Selecting copulas for risk management.
  {Journal of Banking and Finance}   {31}(8), 2405--2423.


\hangindent=\parindent\noindent Lando, D., 2004.   {Credit Risk Modeling: Theory and Applications.} Princeton University Press, Princeton.

\hangindent=\parindent\noindent Ledford, A.W., Tawn, J.A., 1996. Statistics for near independence in multivariate
extreme values.   {Biometrika}   {83}(1), 169--187.

\hangindent=\parindent\noindent Lehmann, E.L., 1966. Some concepts of dependence.   {The Annals of Mathematical
Statistics}   {37}(5), 1137--1153.


\hangindent=\parindent\noindent Li, L., Yuen, K.C., Yang, J., 2014. Distorted mix method for constructing copulas
with tail dependence.   {Insurance: Mathematics and Economics}   {57}, 77--89.

\hangindent=\parindent\noindent Marshall, A.W., Olkin, I., 1967. A generalized bivariate exponential distribution.
  {Journal of Applied Probability}   {4}(2), 291--302.

\hangindent=\parindent\noindent McNeil, A.J., Frey, R., Embrechts, P., 2005. Quantitative Risk Management: Concepts, Techniques, and Tools. Princeton University Press, Princeton.

\hangindent=\parindent\noindent Moschopoulos, P.G., 1985. The distribution of the sum of independent gamma random
variables.   {Annals of the Institute of Statistical Mathematics}   {37}(1), 541--544.

\hangindent=\parindent\noindent Nelsen, R.B., 2006.   {An Introduction to Copulas}, 2nd ed. Springer, New York.






\hangindent=\parindent\noindent Sandstr\"{o}m, A., 2010.   {Handbook of Solvency for Actuaries and Risk Managers: Theory and Practice.} CRC Press, Boca Raton.

\hangindent=\parindent\noindent Schmid, F., Schmidt, R., 2007. Multivariate conditional versions of Spearman's rho
and related measures of tail dependence.   {Journal of Multivariate Analysis}   {98}(6), 1123--1140.


\hangindent=\parindent\noindent Sklar, A., 1959. Fonction de r\'{e}partition \`{a} n dimensions et leurs marges. {Publications de l'Institut de Statistique de l'Universit\'{e} de Paris}   {8}, 229--231.

\hangindent=\parindent\noindent Skoglund, J., Chen, W., 2015.   {Financial Risk Management: Applications in Market, Credit, Asset and Liability Management and Firmwide Risk.} Wiley, Hoboken.

\hangindent=\parindent\noindent Staudt, A., 2010. Tail risk, systemic risk and copulas.   {Casualty Actuarial Society}   {2}, 1--23.

\hangindent=\parindent\noindent Su, J., Furman, E., 2016a. A form of multivariate Pareto distribution with applications to financial risk measurement. ASTIN Bulletin, in press.

\hangindent=\parindent\noindent Su, J., Furman, E., 2016b. Multiple risk factor dependence structures: Distributional properties and applications in actuarial mathematics. Technical report, available at SSRN: http://ssrn.com/abstract=2694308. Accessed on September 1, 2016.

\hangindent=\parindent\noindent Zhang, M.H., 2008. Modelling total tail dependence along diagonals. {Insurance: Mathematics and Economics}   {42}(1), 73--80.

\appendix
\section{Proofs}

\begin{proof}[Proof of Equation \ref{marginal-general-ddf}]
For $t\in\mathbf{R}_{0,+}$ and $i=1,\ldots,n$, we have by construction that
\begin{eqnarray*}
S_i(t):=\mathbf{P}[\tau_i> t]
&=&\mathbf{E}\left[
\mathbf{P}\left[
\sum_{j\in\mathcal{RF}^l_i} N_{\Lambda_jt}+\sum_{j\in\mathcal{RF}^m_i} {}_iN_{\Lambda_jt}=0
\bigg| \boldsymbol{\Lambda}
\right]
\right] \nonumber \\
&=&
\prod_{j\in\mathcal{RF}_i^l} \mathbf{P}[E_{\Lambda_j}>t]
\prod_{j\in\mathcal{RF}_i^m} \mathbf{P}[{}_iE_{\Lambda_j}>t] \nonumber\\
&=&
\prod_{j\in\mathcal{RF}_i^l}  \psi_{{\Lambda_j}}(t)
\prod_{j\in\mathcal{RF}_i^m} \psi_{{\Lambda_j}} (t)=\psi_{{\sum_{j\in\mathcal{RF}_i} \Lambda_j}}(t),
\end{eqnarray*}
which proves the assertion.
\end{proof}

\begin{proof}[Proof of Equation \ref{multivariate-general-ddf}]
By construction and for $t_i\in\mathbf{R}_{0,+},\ i=1,\ldots,n$, we obtain the following string of equations
\begin{eqnarray*}
&&S(t_1,\ldots,t_n):=\mathbf{P}[\tau_1>t_1,\ldots,\tau_n>t_n]\nonumber\\
&=&\mathbf{E}\left[
\mathbf{P}\left[
\sum_{i=1}^n\sum_{j\in\mathcal{RF}^l_i} N_{\Lambda_jt_i}+\sum_{i=1}^n\sum_{j\in\mathcal{RF}^m_i} {}_iN_{\Lambda_jt_i}=0
\bigg| \boldsymbol{\Lambda}
\right]
\right] \nonumber\\
&=&\mathbf{E}\left[
\mathbf{P}\left[
\sum_{j=1}^l \sum_{i\in\mathcal{RC}_j} N_{\Lambda_jt_i}+\sum_{j=l+1}^{l+m} \sum_{i\in\mathcal{RC}_j}  {}_iN_{\Lambda_jt_i}=0
\bigg| \boldsymbol{\Lambda}
\right]
\right]\nonumber\\
&=&\prod_{j=1}^l
\mathbf{P}\left[E_{\Lambda_j}>\bigvee_{i\in\mathcal{RC}_j}t_i \right]\prod_{j=l+1}^{l+m}
\mathbf{P}\left[\bigcap_{i\in\mathcal{RC}_j} {}_iE_{\Lambda_j}>t_i
\right]\nonumber\\
&=&\prod_{j=1}^{l}\psi_{\Lambda_j}\left(\bigvee_{i\in\mathcal{RC}_j}t_i\right)
\prod_{j=l+1}^{l+m}\psi_{\Lambda_j}\left(\sum_{i\in\mathcal{RC}_j}t_i\right),
\end{eqnarray*}
which proves the desired equation.
\end{proof}

\begin{proof}[Proof of Theorem \ref{pro-alternative def1}]
As the equivalence of (D1) and (D2) is trivial, we only prove that (D1) is equivalent to (\ref{Non-hom-defaults}).
By conditioning, we have that, for
$t_i\in\mathbf{R}_{0,+},\ i=1,\ldots,n$,
\begin{eqnarray*}
&&S\left(t_1,\ldots,t_n\bigg|\ \Lambda_j(t_i),N_{\Lambda_j(t_i)},i=1,\ldots,n,j=1,\ldots,l+m\right)\\
&=&\mathbf{P}\left[{}_1E_1>\Theta_{1}(t_1),\ldots,{}_nE_1>\Theta_{n}(t_n)\bigg|\ \Lambda_j(t_i),
N_{\Lambda_j(t_i)},i=1,\ldots,n,j=1,\ldots,l+m\right]\\
&=&\prod_{i=1}^n\exp\left\{-\sum_{j\in \mathcal{RF}^l_i}\mathbf{I}^{\infty}_{\{N_{\Lambda_j(t_i)}>0\}}-\sum_{j\in \mathcal{RF}^m_i}\Lambda_{j}(t_i)\right\}\\
&=&\prod_{j=1}^l \exp\left\{-\sum_{i\in \mathcal{RC}_j} \mathbf{I}^{\infty}_{\{N_{\Lambda_j(t_i)}>0\}} \right\} \prod_{j=l+1}^{l+m}  \exp\left\{-\sum_{i\in \mathcal{RC}_j}\Lambda_{j}(t_i)\right\}.
\end{eqnarray*}
Consequently, the unconditional joint survival function is given, for $t_i\in\mathbf{R}_{0,+},\ i=1,\ldots,n$, by
\begin{eqnarray}
\label{eqn-alternative def1}
&&S(t_1,\ldots,t_n)=\mathbf{E}\left[
S\left(t_1,\ldots,t_n\bigg|\ \Lambda_j(t_i),N_{\Lambda_j(t_i)},i=1,\ldots,n,j=1,\ldots,l+m\right)
\right]\nonumber\\
&=&
\prod_{j=1}^{l}\mathbf{E}
\left[\exp \left\{-
\Lambda_{j}\left(\bigvee_{i\in \mathcal{RC}_j}t_i\right) \right\} \right]
\prod_{j=l+1}^{l+m} \mathbf{E}\left[\exp\left\{-\sum_{i\in\mathcal{RC}_j}\Lambda_{j}(t_i)\right\}\right] ,
\end{eqnarray}
since
\begin{eqnarray*}
&&\mathbf{E}\left[\exp\left\{-\sum_{i\in \mathcal{RC}_j} \mathbf{I}^{\infty}_{\{N_{\Lambda_j(t_i)}>0\}} \right\}\right]=
\mathbf{P}\left[ \bigcap_{i\in \mathcal{RC}_j}\left\{\mathbf{I}^{\infty}_{\{N_{\Lambda_j(t_i)>0}\}}=0\right\}\right]\\
&=&\mathbf{P}
\left[N_{\Lambda_j}\left(\bigvee_{i\in \mathcal{RC}_j}t_i\right)=0 \right]=
\mathbf{E}
\left[\exp \left\{-
\Lambda_{j}\left(\bigvee_{i\in \mathcal{RC}_j}t_i\right) \right\} \right].
\end{eqnarray*}
Finally, by rewriting (\ref{eqn-alternative def1}) in terms of the Laplace transforms of  $\Lambda_j(t)$, we obtain joint survival function (\ref{jointddf}).
This completes the proof.
\end{proof}

\begin{proof}[Proof of Theorem \ref{pro-simulation}]
To prove the `if' part, note that, for $u_i\in [0,\ 1],\ i=1,\ldots,n$, the c.d.f. of $\mathbf{U}=(U_1,\ldots,U_n)'$ is
\begin{eqnarray*}
&&\mathbf{P}[U_1\leq u_1,\ldots,U_n\leq u_n]\\
&=&\mathbf{P}\left[\bigcap_{j\in \mathcal{RF}_i^l}
\left\{\frac{\ln(V_{j})}{\Lambda_j} \leq -\psi_i^{-1}(u_i)\right\}
\textnormal{ and } \bigcap_{j\in \mathcal{RF}_i^m} \left\{\frac{\ln({}_iV_{j})}{\Lambda_j} \leq -\psi_i^{-1}(u_i)\right\}
\textnormal{ for all }i=1,\ldots,n\right]\\
&=& \mathbf{P}\left[
 \bigcap_{j=1}^l \left\{\frac{\ln(V_{j})}{\Lambda_j} \leq -\bigvee_{i\in \mathcal{RC}_j} \psi_i^{-1}(u_i) \right\}\right]
 \mathbf{P}\left[\bigcap_{j=l+1}^{l+m} \bigcap_{i\in \mathcal{RC}_j}\left\{\frac{\ln({}_iV_{j})}{\Lambda_j} \leq -\psi_i^{-1}(u_i)\right\}\right],
\end{eqnarray*}
where
\begin{eqnarray*}
&& \mathbf{P}\left[
 \bigcap_{j=1}^l \left\{\frac{\ln(V_{j})}{\Lambda_j} \leq -\bigvee_{i\in \mathcal{RC}_j} \psi_i^{-1}(u_i) \right\}\right]=
\mathbf{E}\left[\mathbf{P}\left[
 \bigcap_{j=1}^l \left\{\frac{\ln(V_{j})}{\Lambda_j} \leq -\bigvee_{i\in \mathcal{RC}_j} \psi_i^{-1}(u_i) \right\}\bigg | \mathbf{\Lambda} \right]\right]\\
&=&\mathbf{E}\left[\prod_{j=1}^l\exp\left\{-\bigvee_{i\in \mathcal{RC}_j}\Lambda_j\psi_i^{-1}(u_i)  \right\}\right]
=\prod_{j=1}^l\psi_{\Lambda_j}\left(\bigvee_{i\in \mathcal{RC}_j}\psi_i^{-1}(u_i)\right)
\end{eqnarray*}
and similarly
\begin{eqnarray*}
&& \mathbf{P}\left[\bigcap_{j=l+1}^{l+m} \bigcap_{i\in \mathcal{RC}_j}\left\{\frac{\ln({}_iV_{j})}{\Lambda_j} \leq -\psi_i^{-1}(u_i)\right\}\right]=
\mathbf{E}\left[\mathbf{P}\left[\bigcap_{j=l+1}^{l+m} \bigcap_{i\in \mathcal{RC}_j}\left\{\frac{\ln({}_iV_{j})}{\Lambda_j} \leq -\psi_i^{-1}(u_i)\right\}\bigg | \mathbf{\Lambda}\right]\right] \\
&=&\mathbf{E}\left[\prod_{j=l+1}^{l+m} \exp\left\{-\Lambda_j\sum_{i\in\mathcal{RC}_j} \psi_i^{-1}(u_i) \right\}\right]
=\prod_{j=l+1}^{l+m} \psi_{\Lambda_j}\left(\sum_{i\in \mathcal{RC}_j}\psi_i^{-1}(u_i) \right).
\end{eqnarray*}
Hence the joint c.d.f. of $\mathbf{U}$ coincides with (\ref{cop-lin}).  The `only if' part follows by the uniqueness of the
Laplace transform.
This completes the proof.
\end{proof}

\begin{proof}[Proof of Theorem \ref{pro-simultaneous probability}]
Since the non-zero probability of simultaneous default can only come from the risk factors in
$\mathcal{RF}^l_i,\ i=1,\ldots,n$,
we obtain, for any $t\in \mathbf{R}_{0,+}$ and by conditioning on $\boldsymbol{\Lambda}(t):=(\Lambda_{1}(t),\ldots, \Lambda_{l+m}(t))'$, that
\begin{eqnarray*}
&&\mathbf{P}\left[\tau_{i_1}=\cdots=\tau_{i_k} |\ \boldsymbol{\Lambda}(t)\right] \\
&=&\int_{0}^{\infty}\mathbf{P}\left[\bigcap_{i \in \{i_1,\ldots,i_k \}}
\left\{\sum_{j\in \mathcal{RF}^l_{\overline{(i_1,\ldots,i_n)}}} N_{\Lambda_j(t)}+\sum_{j\in \mathcal{RF}^m_i} {}_iN_{\Lambda_j(t)}=0\right\}
\cap  \mathcal{A}(t)  \bigg| \mathbf{\Lambda}(t) \right] \nonumber dt \\
&=&\int_{0}^{\infty} \mathbf{P}\left[\mathcal{A}(t)| \boldsymbol{\Lambda}(t)\right]
\prod_{j\in \mathcal{RF}^l_{\overline{(i_1,\ldots,i_k)}}} \mathbf{P}\left[ N_{\Lambda_j(t)}=0\bigg|\boldsymbol{\Lambda}(t)\right]
\prod_{j\in \mathcal{RF}^m_{{i_1,\ldots,i_k}}}\mathbf{P}\left[\sum_{i\in \mathcal{RC}^m_j} {}_iN_{\Lambda_j(t)}=0\bigg|\boldsymbol{\Lambda}(t)\right]\ dt.
\end{eqnarray*}
The proof is then completed by interchanging the order of integration.
\end{proof}

\begin{proof}[Proof of Corollary \ref{lemma-simultaneous default}]
Under the assumption of linearity, we obviously have that, for $j=1,\ldots,l+m$ and $t\in\mathbf{R}_{0,+}$,
\[
\psi_{\Lambda_j(t)}(1)=\psi_{\Lambda_j}(t)=\mathbf{P}[E_{\Lambda_j}>t],
\]
as well as that
\[
\mathbf{P}[\mathcal{A}(t)]=\mathbf{E}\left[
-\frac{d}{dt} \exp\left\{
-\sum_{j\in\mathcal{RF}^l_{(i_1,\ldots,i_k)}}\Lambda_j t
\right\}
\right]=\mathbf{E}\left[
\exp\left\{
-\sum_{j\in\mathcal{RF}^l_{(i_1,\ldots,i_k)}}\Lambda_j t
\right\}\sum_{j\in\mathcal{RF}^l_{(i_1,\ldots,i_k)}}\Lambda_j
\right].
\]
Consequently, denoting by $L_j$ the c.d.f. of the r.v. $\Lambda_j$,  the integrand in (\ref{thsim}) reduces to
\begin{eqnarray*}
&&\int_{\mathbf{R}_{0,+}^{|\mathcal{RF}_{i_1,\ldots,i_k}|+1}} \sum_{j\in\mathcal{RF}^l_{(i_1,\ldots,i_k)}}\lambda_j\ \exp\left\{-\left(
\sum_{j\in\mathcal{RF}^l_{i_1,\ldots,i_k}}\lambda_j+
\sum_{j\in\mathcal{RF}^m_{i_1,\ldots,i_k}}\lambda_j  | \mathcal{RC}_j|\right)t
\right\}d\prod_{j\in\mathcal{RF}_{i_1,\ldots,i_k}} L_j(\lambda_j) dt \\
&&=\int_{\mathbf{R}_{0,+}^{|\mathcal{RF}_{i_1,\ldots,i_k}|}}
\frac{\sum_{j\in\mathcal{RF}^l_{(i_1,\ldots,i_k)}}\lambda_j}{\sum_{j\in\mathcal{RF}^l_{i_1,\ldots,i_k}}\lambda_j+
\sum_{j\in\mathcal{RF}^m_{i_1,\ldots,i_k}}\lambda_j  | \mathcal{RC}_j|}d\prod_{j\in\mathcal{RF}_{i_1,\ldots,i_k}} L_j(\lambda_j).
\end{eqnarray*}
This completes the proof.
\end{proof}

\begin{proof}[Proof of Proposition \ref{CGS simultaneous default}]
First note that $\Lambda_{c,i_1,\ldots,i_k}^l\sim Ga(\alpha_{c,i_1,\ldots,i_k},1)$ and $\tilde{\Lambda}_{i_1,\ldots,i_k}^m\sim Ga(\gamma_{c,i_1,\ldots,i_k}+K,1)$ appealing to
 Lemma \ref{FL2005}. Then, conditionally on $K=h$, the distribution of
$\Lambda_{c,i_1,\ldots,i_n}^l+\tilde{\Lambda}_{i_1,\ldots,i_k}^m$ is $Ga(\xi_{c,i_1,\ldots,i_k}+h,1)$, and
the conditional probability of simultaneous default
is equal to the expectation of a beta distributed r.v. with parameters
$\alpha_{c,i_1,\ldots,i_k}$ and $\gamma_{c,i_1,\ldots,i_k}+h$, where $h$ is a non-negative real number. The assertion of the
proposition follows evoking the law of iterated expectation.
\end{proof}

\begin{proof}[Proof of Theorem \ref{sp-rho}]
By definition, we have that, for $1\leq i\neq k\leq n$,
\begin{eqnarray}
\label{joint bivariate expectation}
&&(\rho_S(C_{\boldsymbol{\xi}})+3)/{12}=\int\int_{[0,1]^2}u_iu_kdC_{\boldsymbol{\xi}}(u_i,u_k)=\int\int_{[0,1]^2}C_{\boldsymbol{\xi}}(u_i,u_k)du_i du_k \nonumber\\
&=&\int_0^{1}\int_{0}^{u_k^{{\xi_{c,i}}/{\xi_{c,k}}}}u_i^{\frac{\xi_{c,i,\overline{(i,k)}}+\alpha_{c,(i,k)}}{\xi_{c,i}}}
u_k^{\frac{\xi_{c,k,\overline{(i,k)}}}{\xi_{c,k}}}
\left(u_i^{-\frac{1}{\xi_{c,i}}}+u_k^{-\frac{1}{\xi_{c,k}}}-1 \right)^{-\gamma_{c,(i,k)}}
du_i du_k \nonumber\\
&+& \int_0^{1}\int_{u_k^{{\xi_{c,i}}/{\xi_{c,k}}}}^1u_i^{\frac{\xi_{c,i,\overline{(i,k)}}}{\xi_{c,i}}}
u_k^{\frac{\xi_{c,k,\overline{(i,k)}}+\alpha_{c,(i,k)}}{\xi_{c,k}}}
\left(u_i^{-\frac{1}{\xi_{c,i}}}+u_k^{-\frac{1}{\xi_{c,k}}}-1 \right)^{-\gamma_{c,(i,k)}}
du_i du_k \nonumber\\
&=&\int_0^{1}\int_{0}^{u_k^{{\xi_{c,i}}/{\xi_{c,k}}}}u_i^{\frac{\xi_{c,i,\overline{(i,k)}}+\alpha_{c,(i,k)}}{\xi_{c,i}}}
u_k^{\frac{\xi_{c,k,\overline{(i,k)}}}{\xi_{c,k}}}
\left(u_i^{-\frac{1}{\xi_{c,i}}}+u_k^{-\frac{1}{\xi_{c,k}}}-1 \right)^{-\gamma_{c,(i,k)}}
du_i du_k \nonumber\\
&+&\int_0^{1}\int_{0}^{u_i^{{\xi_{c,i}}/{\xi_{c,k}}}}u_i^{\frac{\xi_{c,k,\overline{(i,k)}}+\alpha_{c,(i,k)}}{\xi_{c,k}}}
u_k^{\frac{\xi_{c,i,\overline{(i,k)}}}{\xi_{c,i}}}
\left(u_i^{-\frac{1}{\xi_{c,k}}}+u_k^{-\frac{1}{\xi_{c,i}}}-1 \right)^{-\gamma_{c,(i,k)}}
du_k du_i \notag \\
&=&I_1(\boldsymbol{\xi})+I_2(\boldsymbol{\xi}).
\end{eqnarray}
We further compute $I_1(\boldsymbol{\xi})$ whereas the other integral can be tackled in
a similar fashion. By change of variables and evoking Equation $(3.197(1))$ in Gradshteyn and Ryzhik (2014), we obtain that
\begin{eqnarray*}
\label{i1}
&&I_1(\boldsymbol{\xi})\\
&=&\xi_{c,i}\xi_{c,k}\int_{\mathbf{R}_{0,+}}(1+x)^{-\xi_{c,k}-\xi_{c,k,\overline{(i,k)}}-1}\int_{\mathbf{R}_{0,+}}\left(1+2x+y \right)^{-\gamma_{c,(i,k)}}\left(1+x+y \right)^{{-\xi_{c,i}-\xi_{c,i,\overline{(i,k)}}-\alpha_{c,(i,k)}}-1}dy dx\nonumber\\
&=&\frac{\xi_{c,k}}{2}\int_{\mathbf{R}_{0,+}}
(1+2x)^{-\gamma_{c,(i,k)}}
(1+x)^{-(b+1-\gamma_{c,(i,k)})}{}_2F_1\left(
\gamma_{c,(i,k)},1;2\xi_{c,i}+1;\frac{x}{1+2x}
\right)
dx\\
&=& \frac{\xi_{c,k}}{4}\int_0^{1}(1-v)^{b-1}\left( 1-v/2 \right)^{-(b+1-\gamma_{c,(i,k)})} \ _2F_1\left(\gamma_{c,(i,k)},1;2\xi_{c,i}+1;v/2\right) dv,
\end{eqnarray*}
where $b=2\xi_{c,i}+2\xi_{c,k}-\xi_{c,(i,k)}$.
Furthermore, note that as the ${}_2F_1$ hypergeometric function has the following integral representation for all $v\in\mathbf{R}$,
\[
\frac{1}{2\xi_{c,i}}\ _2F_1\left(\gamma_{c,(i,k)},1;2\xi_{c,i}+1;v/2\right)=\int_0^1 (1-t)^{2\xi_{c,i}-1} \left(1-\frac{v}{2}t\right)^{-\gamma_{c,(i,k)}}dt
\]
(Equation 9.111 in loc. cit.), we obtain the following string of integrals
\begin{eqnarray*}
I_1(\boldsymbol{\xi})&=&\frac{\xi_{c,i} \xi_{c,k}}{2}\int_0^1 (1-t)^{2\xi_{c,i}-1} \int_0^1 \left(1-v\right)^{b-1} \left( 1-\frac{v}{2} \right)^{-(b+1-\gamma_{c,(i,k)})}\left(1-\frac{t}{2}v\right)^{-\gamma_{c,(i,k)}} dv dt\\
&\overset{(1)}{=}& \frac{\xi_{c,i} \xi_{c,k}}{2b} \int_0^1 (1-t)^{2\xi_{c,i}-1} F_1\left(1,b+1-\gamma_{c,(i,k)},\gamma_{c,(i,k)},
b+1;1/2,t/2\right)dt\\
&\overset{(2)}{=}& \frac{\xi_{c,i} \xi_{c,k}}{b} \int_0^1 (1-t)^{2\xi_{c,i}-1} \ _2F_1\left(1,\gamma_{c,(i,k)};b+1;t-1\right)dt\\
&=& \frac{\xi_{c,i} \xi_{c,k}}{b} \int_0^1 y^{2\xi_{c,i}-1} \ _2F_1\left(1,\gamma_{c,(i,k)};b+1;-y\right)dy\\
&\overset{(3)}{=}& \frac{\xi_{c,k}}{2b} \ _3F_2\left(2\xi_{c,i},1,\gamma_{c,(i,k)};2\xi_{c,i}+1,b+1;-1\right),
\end{eqnarray*}
where
$F_1$ is the bivariate hypergeometric function,
and  `$\overset{(1)}{=}$', `$\overset{(2)}{=}$' and `$\overset{(3)}{=}$' hold by
Equations (3.211), (9.182(1)) and (7.512(12)), respectively, in Gradshteyn and Ryzhik (2014).
The expression for $I_2(\boldsymbol{\xi})$ is then by analogy
\begin{eqnarray*}
&&I_2(\boldsymbol{\xi})= \frac{\xi_{c,i}}{2b} \ _3F_2\left(2\xi_{c,k},1,\gamma_{c,(i,k)};2\xi_{c,k}+1,b+1;-1\right).
\end{eqnarray*}
We note in passing that the hypergeometric functions in $I_1(\boldsymbol{\xi})$ and  $I_2(\boldsymbol{\xi})$
converge absolutely since  $b+1-\gamma_{c,(i,k)}>1$ for
$1\leq i\neq k\leq n$.
This completes the proof.
\end{proof}

\begin{proof}[Proof of Corollary \ref{MO-spearman-thm}]
The assertion follows since
${}_3F_2(a,b,0;c,d;z)\equiv 1$, for any real $a,b,c,d,z$.
\end{proof}

\begin{proof}[Proof of Corollary \ref{pro-as-spearman}]
First notice that according to Theorem 3.4.1. in Su and Furman 
(2016b), we have that
\begin{eqnarray}
\label{combine-hyper}
&&\frac{1}{s-2}\left(({s_2-1})\ _3F_2(s_1-1,1,a;s_1,s-1;-1)+{(s_1-1)}\ _3F_2(s_2-1,1,a;s_2,s-1;-1)\right) \nonumber\\
&=&\ _3F_2(a,1,1;s_1,s_2;1),
\end{eqnarray}
where $a,b,c$ are all positive and such that $s_1=a+b>2$, $s_2=a+c>2$, and $s=a+b+c$.  Put $\alpha_{c,(i,k)}\equiv 0$, $a=2\xi_{c,(i,k)}=2\gamma_{c,(i,k)}$, $b=2\xi_{c,i,\overline{(i,k)}}+1$ and $c=2\xi_{c,k,\overline{(i,k)}}+1$, then we have that $s_1=2\xi_{c,i}+1$ and $s_2=2\xi_{c,k}+1$,
and the assertion follows using (\ref{sp-rho-ClaytonMRF-formula}).
\end{proof}

\begin{proof}[Proof of Proposition \ref{pro-strong tail dependence}]
By (\ref{bivariate-ddf}), we have the limit
\begin{eqnarray}
\label{strong-tail-formula1}
\lambda_L(C_{\boldsymbol{\xi}})&=&\lim_{u\downarrow 0}\frac{u^{\frac{\xi_{c,i,\overline{(i,k)}}}{\xi_{c,i}}+\frac{\xi_{c,k,\overline{(i,k)}}}{\xi_{c,k}}+\frac{\alpha_{c,(i,k)}}
{\xi_{c,i}\wedge \xi_{c,k}}}\left(u^{-\frac{1}{\xi_{c,i}}}+u^{-\frac{1}
{\xi_{c,k}}}-1 \right)^{-\gamma_{c,(i,k)}}}{u}.
\end{eqnarray}

First consider the case when
$\xi_{c,i,\overline{(i,k)}}\neq 0$ and/or $\xi_{c,k,\overline{(i,k)}} \neq0$, and set without loss
of generality $\xi_{c,i}<\xi_{c,k}$. Then the limit becomes
\begin{eqnarray*}
\lambda_L(C_{\boldsymbol{\xi}})&=&\lim_{u\downarrow 0}\frac{u^{\frac{\xi_{c,i,\overline{(i,k)}}+\alpha_{c,(i,k)}}{\xi_{c,i}}+\frac{\xi_{c,k,\overline{(i,k)}}}{\xi_{c,k}}}\left(u^{-\frac{1}{\xi_{c,i}}}+u^{-\frac{1}
{\xi_{c,k}}}-1 \right)^{-\gamma_{c,(i,k)}}}{u}\\
&=&\lim_{u\downarrow 0}{u^{-\frac{\gamma_{c,(i,k)}}{\xi_{c,i}}+\frac{\xi_{c,k,\overline{(i,k)}}}{\xi_{c,k}}}\left(u^{-\frac{1}{\xi_{c,i}}}+u^{-\frac{1}
{\xi_{c,k}}}-1 \right)^{-\gamma_{c,(i,k)}}}\\
&=&\lim_{u\downarrow 0}{u^{\frac{\xi_{c,k,\overline{(i,k)}}}{\xi_{c,k}}}\left(1+u^{-\frac{1}
{\xi_{c,k}}+\frac{1}{\xi_{c,i}}}-u^{\frac{1}{\xi_{c,i}}} \right)^{-\gamma_{c,(i,k)}}}=0.
\end{eqnarray*}

In the other case, i.e., when both $\xi_{c,i,\overline{(i,k)}}$ and $\xi_{c,k,\overline{(i,k)}}$ are zero, we have that that $\xi_{c,i}=\xi_{c,k}=\gamma_{c,(i,k)}+\alpha_{c,(i,k)}$,
where $1\leq i\neq k\leq n$, hence limit (\ref {strong-tail-formula1}) becomes
\begin{eqnarray*}
\lambda_L(C_{\boldsymbol{\xi}})&=&\lim_{u\downarrow 0}\frac{u^{\frac{\alpha_{c,(i,k)}}{\xi_{c,i}}}\left(u^{-\frac{1}{\xi_{c,i}}}+u^{-\frac{1}
{\xi_{c,i}}}-1 \right)^{-\gamma_{c,(i,k)}}}{u}\\
&=&\lim_{u\downarrow 0}{u^{-\frac{\gamma_{c,(i,k)}}{\xi_{c,i}}}\left(2u^{-\frac{1}{\xi_{c,i}}}-1 \right)^{-\gamma_{c,(i,k)}}}\\
&=&\lim_{u\downarrow 0}{\left(2-u^{\frac{1}
{\xi_{c,i}}} \right)^{-\gamma_{c,(i,k)}}}=2^{-\gamma_{c,(i,k)}}.\\
\end{eqnarray*}
This completes the proof.
\end{proof}

\begin{proof}[Proof of Proposition \ref{pro-two indices}]
We only need to prove (\ref{kl-genCl}), as the other formula follows from the relationship
$\chi_L(C)=2/\kappa_L(C)-1$. Then, for $\xi_{c,i}< \xi_{c,k}$ and $1\leq i\neq k\leq n$, we
 have, by  (\ref{bivariate-ddf}) and for $u\in(0,\ 1)$, that
\begin{eqnarray*}
C_{\boldsymbol{\xi}}(u,u)&=&u^{\frac{\xi_{c,i,\overline{(i,k)}}+\alpha_{c,(i,k)}}{\xi_{c,i}}+\frac{\xi_{c,k,\overline{(i,k)}}}{\xi_{c,k}}}\left(u^{-\frac{1}{\xi_{c,i}}}+u^{-\frac{1}{\xi_{c,k}}}-1\right)^{-\gamma_{c,(i,k)}}\\
&=&u^{\frac{\xi_{c,i,\overline{(i,k)}}+\gamma_{c,(i,k)}+\alpha_{c,(i,k)}}{\xi_{c,i}}+\frac{\xi_{c,k,\overline{(i,k)}}}{\xi_{c,k}}}\left(1+u^{-\frac{1}{\xi_{c,k}}+\frac{1}{\xi_{c,i}}}-u^{\frac{1}{\xi_{c,i}}}\right)^{-\gamma_{c,(i,k)}} \\
&=&u^{1+\frac{\xi_{c,k,\overline{(i,k)}}}{\xi_{c,k}}}\left(1+u^{-\frac{1}{\xi_{c,k}}+\frac{1}{\xi_{c,i}}}-u^{\frac{1}{\xi_{c,i}}}\right)^{-\gamma_{c,(i,k)}},
\end{eqnarray*}
which yields $\kappa_L(C_{\boldsymbol{\xi}})=1+\frac{\xi_{c,k,\overline{(i,k)}}}{\xi_{c,k}}$.
For $\xi_{c,i}\geq \xi_{c,k}$, we have by analogy that
 $\kappa_L(C_{\boldsymbol{\xi}})=1+\frac{\xi_{c,i,\overline{(i,k)}}}{\xi_{c,i}}$, which
 establishes  (\ref{kl-genCl}).  This completes the proof.
\end{proof}

\begin{proof}[Proof of Theorem \ref{pro-general mdp}]
Assume without loss of generality that $\xi_{c,i}>\xi_{c,k},\ 1\leq i\neq k\leq n$, which is the case when the singularity curve of the
Clayton MRF copula lies in the upper left section of its domain of definition. Also, for $1\leq i\neq k\leq n$, let
\[
\delta_i=\frac{\xi_{c,i,\overline{(i,k)}}+\alpha_{c,(i,k)}}{\xi_{c,i}}-\frac{\xi_{c,k,\overline{(i,k)}}}{\xi_{c,k}}
\textnormal{ and }
\delta_k=\frac{\xi_{c,i,\overline{(i,k)}}}{\xi_{c,i}}-\frac{\xi_{c,k,\overline{(i,k)}}+\alpha_{c,(i,k)}}{\xi_{c,k}},
\]
then, by (\ref{bivariate-ddf}), we have that
\begin{eqnarray}
\label{MDP-copula function1}
&&C_{\boldsymbol{\xi}}\left(x,{u^2}/{x}\right) \\
&=&\left\{
\begin{array}{ll}
x^{\delta_i}u^
{2(1-\xi_{c,(i,k)}/\xi_{c,k})}
\left(x^{-\frac{1}{\xi_{c,i}}}+\left(\frac{u^2}{x}\right)^{-\frac{1}{\xi_{c,k}}}-1 \right)^{-\gamma_{c,(i,k)}},& x\leq
u^{{2\xi_{c,i}}/{(\xi_{c,i}+\xi_{c,k})}}\\
x^{\delta_k}u^
{2(1-\gamma_{c,(i,k)}/\xi_{c,k})}
\left(x^{-\frac{1}{\xi_{c,i}}}+\left(\frac{u^2}{x}\right)^{-\frac{1}{\xi_{c,k}}}-1 \right)^{-\gamma_{c,(i,k)}},& x>u^{{2\xi_{c,i}}/{(\xi_{c,i}+\xi_{c,k})}}
\end{array}
\right., \notag
\end{eqnarray}
and we are interested in the behaviour of (\ref{MDP-copula function1}) on the interval $[u^2,\ 1]$, which is to this end split into two intervals
$[u^2,\ u^{{2\xi_{c,i}}/{(\xi_{c,i}+\xi_{c,k})}})$ and $[u^{{2\xi_{c,i}}/{(\xi_{c,i}+\xi_{c,k})}},\ 1]$ with $u\in(0,\ 1)$.

For $x\in \left[u^2,u^{{2\xi_{c,i}}/{(\xi_{c,i}+\xi_{c,k})}}\right)$, we first note that
\[
\frac{\partial}{\partial x}\ln \left(C_{\boldsymbol{\xi}}(x,u^2/x)\right)=0
\]
if and only if
\begin{eqnarray*}
\label{zeta-i}
\zeta_i(x):=\left(\delta_i+\frac{\gamma_{c,(i,k)}}{\xi_{c,i}} \right) x^{-{1}/{\xi_{c,i}}}+\left(\delta_i-\frac{\gamma_{c,(i,k)}}{\xi_{c,k}}  \right)\left(\left(\frac{u^2}{x}\right)^{-{1}/{\xi_{c,k}}}-1\right)-
\frac{\gamma_{c,(i,k)}}{\xi_{c,k}}=0
\end{eqnarray*}
or, equivalently, if and only if
\begin{equation}\label{eta-i}
\eta_i(x):=\zeta_i(x)x^{-\frac{1}{\xi_{c,k}}}=0.
\end{equation}
Equation (\ref{eta-i}) does not have solutions for  $x\in\left(u^2,u^{{2\xi_{c,i}}/{(\xi_{c,i}+\xi_{c,k})}}\right)$
since $\eta_i(x)$ is non-increasing therein
\begin{eqnarray*}
\eta_i'(x)
&=&{x^{-\frac{1}{\xi_{c,i}}-\frac{1}{\xi_{c,k}}-1}}
\left(- \left(\delta_i+\frac{\gamma_{c,(i,k)}}{\xi_{c,i}} \right)\left(\frac{1}{\xi_{c,i}}+\frac{1}{\xi_{c,k}}\right)+\frac{\delta_i}{\xi_{c,k}} x^{\frac{1}{\xi_{c,i}}} \right)\\
&\leq&-
{x^{-\frac{1}{\xi_{c,i}}-\frac{1}{\xi_{c,k}}-1}}
\frac{\delta_i}{\xi_{c,i}}
\leq 0
\end{eqnarray*}
and such that
 \begin{eqnarray*}
 \eta_i(u^2)&=&u^{-\frac{2}{\xi_{c,k}}}\left(\left(\delta_i+\frac{\gamma_{c,(i,k)}}{\xi_{c,i}} \right) u^{-{2}/{\xi_{c,i}}}-
\frac{\gamma_{c,(i,k)}}{\xi_{c,k}}\right)>0; \\
\eta_i(u^{{2\xi_{c,i}}/{(\xi_{c,i}+\xi_{c,k})}})&=&u^{-\frac{2\xi_{c,i}}{\xi_{c,k}(\xi_{c,i}+\xi_{c,k})}}
\left(\left(2\delta_i+ \frac{\gamma_{c,(i,k)}}{\xi_{c,i}}-\frac{\gamma_{c,(i,k)}}{\xi_{c,k}} \right)u^{-\frac{2}{\xi_{c,i}+\xi_{c,k}}}-\delta_i\right)\\
&=&u^{-\frac{2\xi_{c,i}}{\xi_{c,k}(\xi_{c,i}+\xi_{c,k})}}\left(\left(1+\delta_i-\frac{\gamma_{c,(i,k)}}{\xi_{c,k}} \right)u^{-\frac{2}{\xi_{c,i}+\xi_{c,k}}}-\delta_i\right)>0.
\end{eqnarray*}
Hence, we conclude that $x\mapsto C_{\boldsymbol{\xi}}(x,u^2/x)$ is strictly increasing on
$(u^2,\ u^{{2\xi_{c,i}}/({\xi_{c,i}+\xi_{c,k}})})$ and cannot attain its maximum or
maxima there.

Let us now turn to $x\in [u^{{2\xi_{c,i}}/({\xi_{c,i}+\xi_{c,k}})},1]$. We note that
\begin{eqnarray*}
\frac{\partial}{\partial x}\ln \left(C_{\boldsymbol{\xi}}(x,u^2/x)\right)=0
\end{eqnarray*}
if and only if
\begin{eqnarray*}
\label{zeta-j}
\zeta_k(x):=\left(\delta_k+\frac{\gamma_{c,(i,k)}}{\xi_{c,i}} \right) \left(x^{-{1}/{\xi_{c,i}}}-1\right)+\left(\delta_k-\frac{\gamma_{c,(i,k)}}{\xi_{c,k}}  \right)\left(\frac{u^2}{x}\right)^{-{1}/{\xi_{c,k}}}+\frac{\gamma_{c,(i,k)}}{\xi_{c,i}}=0
\end{eqnarray*}
if and only if
\begin{equation}
\label{eta-k}
\eta_k(x):=\zeta_k(x)x^{\frac{1}{\xi_{c,i}}}=0.
\end{equation}
Equation (\ref{zeta-j}) may have at most one solution for
$x\in (u^{{2\xi_{c,i}}/({\xi_{c,i}+\xi_{c,k}})},1)$
and $u\in(0,\ 1)$, as
\begin{eqnarray*}
\eta_k(1)&=&-\left(\frac{\gamma_{c,(i,k)}+\alpha_{c,(i,k)}}{\xi_{c,i}}\right)u^{-2/\xi_{c,k}}
+\frac{\gamma_{c,(i,k)}}{\xi_{c,i}}<0; \\
 \eta_k'(x)
&=&
{x^{\frac{1}{\xi_{c,i}}-1}}\left(\left(\delta_k- \frac{\gamma_{c,(i,k)}}{\xi_{c,k}} \right)\left(\frac{1}{\xi_{c,i}}+\frac{1}{\xi_{c,k}}\right)
\left(
\frac{u^2}{x}
\right)^{-\frac{1}{\xi_{c,k}}} -\frac{\delta_k}{\xi_{c,k}} \right) \\
&=&
{x^{\frac{1}{\xi_{c,i}}-1}}\left(\left(-\frac{\xi_{c,(i,k)}}{\xi_{c,i}} \right)\left(\frac{1}{\xi_{c,i}}+\frac{1}{\xi_{c,k}}\right)
\left(
\frac{u^2}{x}
\right)^{-\frac{1}{\xi_{c,k}}} -\frac{\delta_k}{\xi_{c,k}} \right)
\leq 0;
\end{eqnarray*}
and the sign of
\begin{equation}
\eta_k\left(u^{2\xi_{c,i}/(\xi_{c,i}+\xi_{c,k})}\right)=\left(\delta_k-\frac{
\alpha_{c,(i,k)}}{\xi_{c,i}} \right)-\delta_ku^{\frac{2}{\xi_{c,i}+\xi_{c,k}}}\\
\end{equation}
is unknown.  Consequently, we have that the function $x\mapsto C_{\boldsymbol{\xi}}(x,u^2/x)$ may or may not achieve its maximum on the interval $(u^{{2\xi_{c,i}}/({\xi_{c,i}+\xi_{c,k}})},1)$, and there may be one such maximum, only.

To summarize, there are two possibilities:

\begin{figure}[t!]
\centering
\includegraphics[width=7.5cm,height=7.5cm]{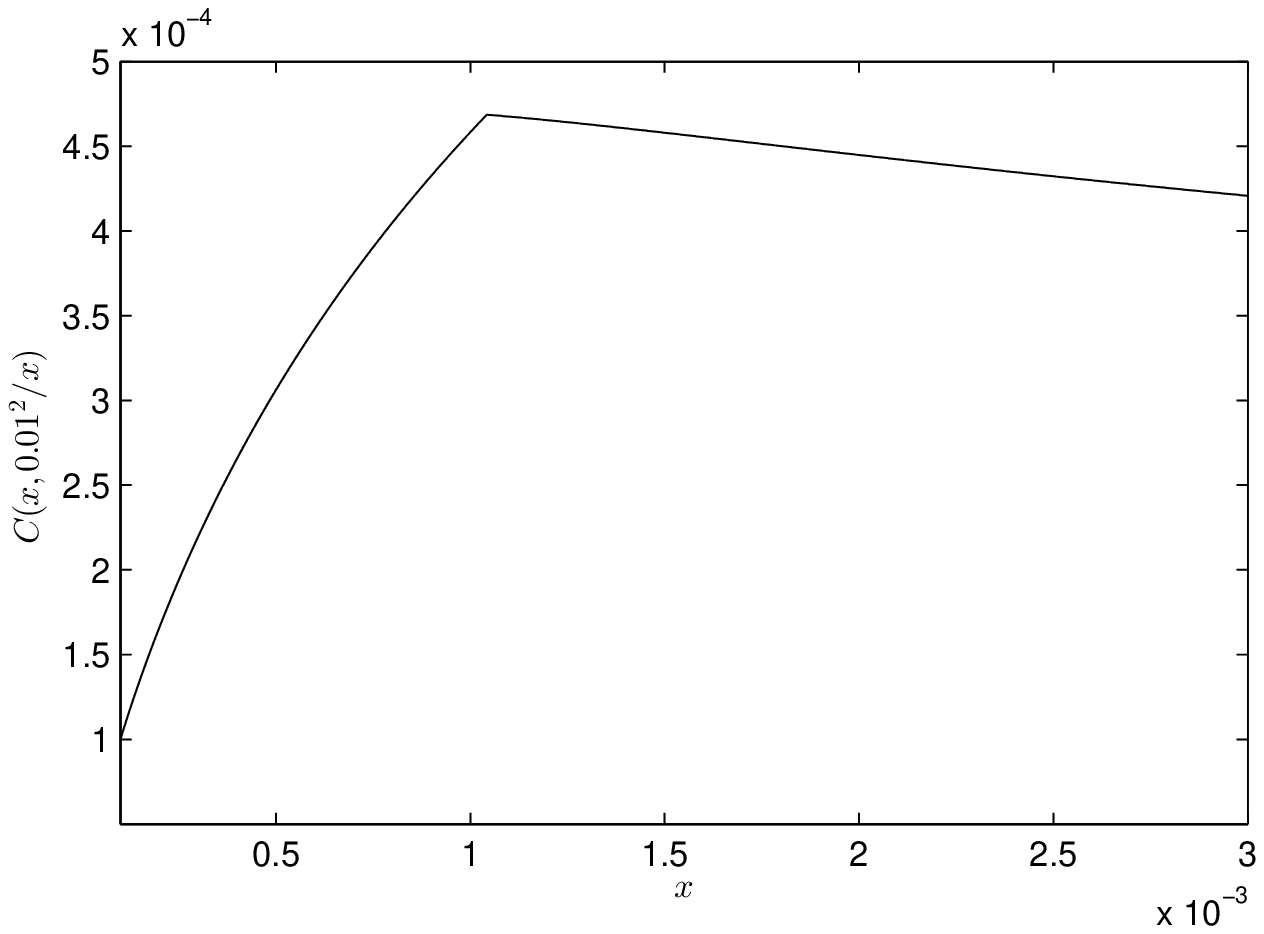}
\hfill
\includegraphics[width=7.5cm,height=7.5cm]{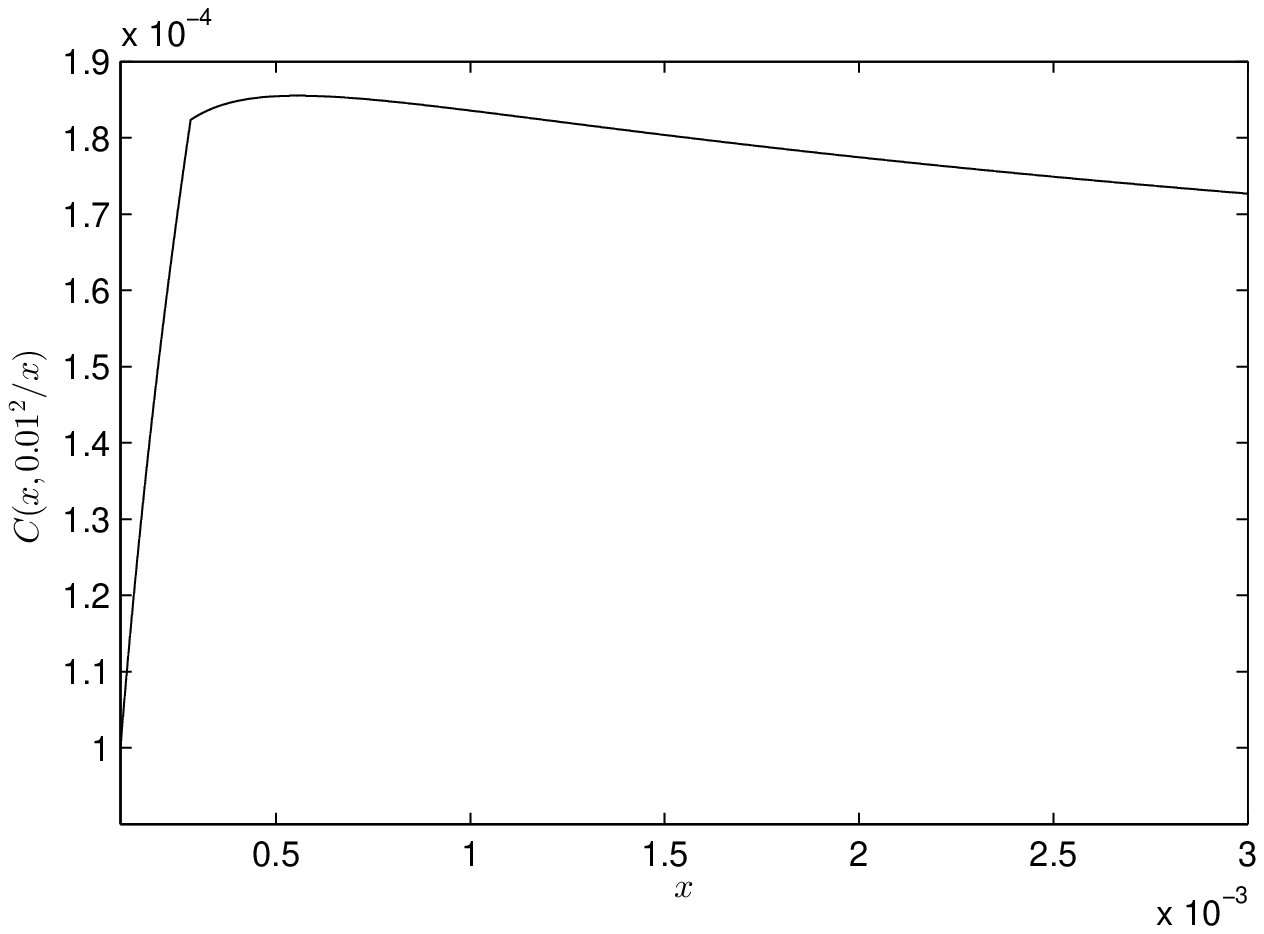}
\caption{The function $C(x,0.01^2/x)$ for $\xi_{c,i,\overline{(i,k)}}=3,\ \xi_{c,k,\overline{(i,k)}}=0.3,\ \gamma_{c,(i,k)}=0.5,\  \alpha_{c,(i,k)}=0.6$ (left panel) and $\xi_{c,i,\overline{(i,k)}}=10,\ \xi_{c,k,\overline{(i,k)}}=0.3,\ \gamma_{c,(i,k)}=0.5,\   \alpha_{c,(i,k)}=0.6$ (right panel).}
\label{fig-mdp-function}
\end{figure}

(1) (Figure \ref{fig-mdp-function}, left panel)
- the function
$x\mapsto C_{\boldsymbol{\xi}}(x,u^2/x)$ is strictly increasing on
$(u^2,\ u^{{2\xi_{c,i}}/({\xi_{c,i}+\xi_{c,k}})})$ and  strictly decreasing on
$(u^{{2\xi_{c,i}}/({\xi_{c,i}+\xi_{c,k}})},\ 1)$. Therefore its maximum is achieved
at $x=u^{{2\xi_{c,i}}/({\xi_{c,i}+\xi_{c,k}})}$, the function of maximal dependence
is $\varphi^\ast(u)=u^{{2\xi_{c,i}}/({\xi_{c,i}+\xi_{c,k}})}$ and the
path of maximal dependence is $(u^{{2\xi_{c,i}}/({\xi_{c,i}+\xi_{c,k}})},\
u^{{2\xi_{c,k}}/({\xi_{c,i}+\xi_{c,k}})})$, where $u\in [0,\ 1]$.
Also, the indices $\lambda_L^\ast$,
$\kappa_L^\ast$ and $\chi_L^\ast$ follow, respectively, from (\ref{lambdaast}),
(\ref{chaiast}) and (\ref{kappaast}).

(2) (Figure \ref{fig-mdp-function}, right panel) - the function
 $x\mapsto C_{\boldsymbol{\xi}}(x,u^2/x)$ has its maximum
on $(u^{{2\xi_{c,i}}/({\xi_{c,i}+\xi_{c,k}})},\ 1)$. In this case, we cannot
formulate the function of maximal tail dependence explicitly, and so the path of maximal
tail dependence is unknown. Nevertheless, the indices of maximal tail dependence can
be written in a closed form. In this respect, we know that the function of maximal
dependence exists, is unique and satisfies the equation $\zeta_k(x)=0$, or equivalently
\begin{equation}\label{solution-2}
x=u^{{2\xi_{c,i}}/{(\xi_{c,i}+\xi_{c,k})}}r(x),
\end{equation}
where
 \[
r(x)=\left(\left(\delta_k-\frac{\gamma_{c,(i,k)}}{\xi_{c,k}} \right) \Big/\left(\delta_kx^{\frac{1}{\xi_{c,i}}} -\left(\delta_k+\frac{\gamma_{c,(i,k)}}{\xi_{c,i}}\right) \right)\right)^{-\frac{\xi_{c,i}\xi_{c,k}}{\xi_{c,i}+\xi_{c,k}}}.
\]
Then the substitution of  (\ref{solution-2}) into  (\ref{MDP-copula function1}) yields
\begin{eqnarray*}
C_{\boldsymbol{\xi}}\left(\varphi^*(u),{u^2}/{\varphi^*(u)}\right) =u^{1+(\xi_{c,i,\overline{(i,k)}}+\xi_{c,k,\overline{(i,k)}})/(\xi_{c,i}+\xi_{c,k})}s(u),
\end{eqnarray*}
where the function $s(u)$ is such that
$\lim _{u\downarrow 0} s(u)=const(\in\mathbf{R}_{0,+})$.
The index of maximal weak lower tail dependence is obtained from the relationship $\xi^*_L=2/\kappa^*_L-1$.  Finally,  the index of maximal strong lower tail
dependence is non-zero if and only if $\kappa^*_L(C_{\boldsymbol{\xi}})\equiv1$, which in
the context of the Clayton MRF copulas implies exchangeability as in this case
$\xi_{c,i,\overline{(i,k)}}\equiv 0$ and  $\xi_{c,k,\overline{(i,k)}}\equiv 0$.
This completes the proof.
\end{proof}

\end{document}